\documentclass[10pt,journal]{IEEEtran}
\usepackage[final]{graphicx}
\usepackage{subfigure}
\usepackage{float}
\usepackage{amsmath}
\usepackage{cases}
\usepackage{color}
\usepackage{colortbl}
\usepackage{algorithm}
\usepackage{algpseudocode}
\usepackage{amsmath}
\usepackage{amssymb}
\usepackage{ulem}
\usepackage{booktabs}
\usepackage{setspace}
\usepackage{array}
\newtheorem{theorem}{\noindent \bf Theorem}
\newenvironment{proof}{{ \noindent \it Proof.}}{\hfill $\blacksquare$}
\usepackage{changepage}

\makeatletter
\renewcommand{\maketag@@@}[1]{\hbox{\m@th\normalsize\normalfont#1}}%
\makeatother
\usepackage{stfloats}
\usepackage{cite}
\usepackage{makecell}
\usepackage{multirow}
\usepackage{hyperref}

\ifCLASSINFOpdf
\else
\fi

\usepackage{caption}
\usepackage{mathtools}

\begin{document}
\title{Integrated Sensing and Communication Enabled Cooperative Passive Sensing Using Mobile Communication System}
\author{
Zhiqing~Wei,~\IEEEmembership{Member,~IEEE,}
Haotian~Liu,~\IEEEmembership{Graduate Student Member,~IEEE,}
Hujun~Li,\\
Wangjun~Jiang,~\IEEEmembership{Graduate Student Member,~IEEE,}
Zhiyong~Feng,~\IEEEmembership{Senior Member,~IEEE,}
Huici~Wu,~\IEEEmembership{Member,~IEEE,}
Ping~Zhang,~\IEEEmembership{Fellow,~IEEE}

\thanks{This work was supported in part by the National Natural Science Foundation of China (NSFC) under Grant 62271081 and Grant U21B2014; 
in part by the Fundamental Research
Funds for the Central Universities under Grant 2024ZCJH01; 
and in part by the National Key
Research and Development Program of China under Grant 2020YFA0711302. This paper has published in IEEE Transactions on Mobile Computing (TMC), doi:10.1109/TMC.2024.3514113.}

\thanks{Zhiqing Wei, Haotian Liu, Wangjun Jiang, Zhiyong Feng, Huici Wu, and Ping Zhang are with Beijing University of Posts and Telecommunications, Beijing 100876, China (emails: \{weizhiqing; haotian\_liu; jiangwangjun; fengzy; dailywu; pzhang\}@bupt.edu.cn). \textit{Corresponding Authors:} Haotian Liu, Zhiqing Wei.

Hujun Li is with China Telecom Co., Ltd. Sichuan Branch, 
Chengdu 610031, China (email: 19108092304@189.cn).}}
\maketitle

\begin{abstract}
Integrated sensing and communication (ISAC) is a potential technology of the 
sixth-generation (6G) mobile communication system, 
which enables communication base station (BS) with sensing capability. 
However, the performance of single-BS sensing is limited, 
which can be overcome by multi-BS cooperative sensing. 
There are three types of multi-BS cooperative sensing, 
including cooperative active sensing, cooperative passive sensing, 
and cooperative active and passive sensing, 
where the multi-BS cooperative passive sensing has 
the advantages of low hardware modification cost and large sensing coverage. 
However, multi-BS cooperative passive sensing faces the challenges of 
synchronization offset mitigation and sensing information fusion. 
To address these challenges, a non-line of sight (NLoS) and line of sight (LoS) 
signal cross-correlation (NLCC) method is proposed to mitigate carrier frequency offset (CFO) and time offset (TO). 
Besides, a symbol-level fusion method of multi-BS sensing information is proposed. 
The discrete samplings of echo signals from multiple BSs are matched independently and coherently accumulated to improve sensing accuracy. 
Moreover, a low-complexity joint angle-of-arrival (AoA) and angle-of-departure (AoD) estimation method is proposed to reduce the computational complexity. 
Simulation results show that symbol-level multi-BS cooperative passive sensing scheme has an order of magnitude higher sensing accuracy than single-BS passive sensing. 
This work provides a reference for the research on multi-BS cooperative passive sensing.
\end{abstract}
\begin{IEEEkeywords}
Cooperative passive sensing, 
carrier frequency offset (CFO),
integrated sensing and communication (ISAC), 
ISAC signal processing, 
joint angle-of-arrival (AoA) and 
angle-of-departure (AoD) estimation, 
multi-BS cooperative sensing,
symbol-level fusion,
time offset (TO).
\end{IEEEkeywords}

\IEEEpeerreviewmaketitle

\section{Introduction}

\subsection{Background and Motivations}
The commonalities between wireless communication and radar sensing in terms of system design, hardware and frequency resources lay the foundations for integrated sensing and communication (ISAC) \cite{wei2023carrier,ISAC_jwj_1,ISAC_lht}.
Recently, ISAC has become one of the key potential technologies 
of the sixth-generation (6G) mobile communication system \cite{10273259,dou2024integrated,ni2021uplink,liu2022integrated,wei2023integrated}.
The 6G mobile communication system is expected to 
facilitate the implementation of 
new applications, such as smart home, smart transportation and 
Internet of vehicles (IoV) \cite{xiao2023novel,li2023reliability,li2023iq,wei2024deep}.
These emerging applications urgently need high-accuracy sensing 
and high-capacity communication. 
%
%
Multi-BS cooperative sensing is a feasible method to provide high-accuracy sensing service \cite{zhang2021enabling, ISAC_jwj_2, liu2024carrier}.

According to the modes of sensing, 
multi-BS cooperative sensing can be categorized into 
multi-BS cooperative active sensing, 
multi-BS cooperative passive sensing, 
and multi-BS cooperative active and passive sensing \cite{ISAC_jwj_1,wei2024deep}. 
Among the three types of cooperative sensing,
multi-BS cooperative passive sensing has the following advantages.
\begin{enumerate}
\item
\textbf{Low hardware modification cost:} 
The signal transmission process of passive sensing is similar to the communication process, 
which does not require the BS to transmit and receive signals simultaneously.
Therefore, passive sensing can reuse 
the hardware equipment of communication system, 
greatly reducing the cost of hardware modification~\cite{wei2024deep,li2019residual}.
\item
\textbf{Large sensing coverage:}
Passive sensing provides the advantages of continuous, 
all-encompassing reception of external radiation sources, 
eliminating the need for beam switching and beam scanning. 
Meanwhile, passive sensing offers a broader coverage of 
sensing compared with active sensing~\cite{xianrong2020research}.
\end{enumerate}

However, there are the following two main challenges in multi-BS cooperative passive sensing.
\begin{enumerate}
    \item \textbf{Synchronization offsets mitigation:} 
    Since the transmitter (Tx) and receiver (Rx) of 
    passive sensing are separated, 
    imperfect synchronization will cause time offsets (TOs) and carrier frequency offsets (CFOs), which further deteriorates the accuracy of range and velocity estimation~\cite{zhang2022integration}. 
    \item \textbf{Sensing information fusion:} 
    The current fusion methods of passive sensing information seldom consider orthogonal frequency division multiplying (OFDM) signal models, and it is difficult to achieve absolute velocity estimation of target.
\end{enumerate}

\begin{table*}[!ht]
\centering
\caption{A summary for the related work, with the abbreviations CS: compressed sensing, MAP: maximum a posteriori, CACC: cross-antenna cross-correlation, CASR: cross antenna signal ratio.}
\label{tab_0}
\resizebox{0.9\textwidth}{!}{%
\renewcommand{\arraystretch}{1.2}
\begin{tabular}{|c|c|c|c|c|}
\hline
\textbf{Challenge} & \textbf{Field or Type} & \multicolumn{1}{c|}{\textbf{\begin{tabular}[c]{@{}c@{}}Related\\ work\end{tabular}}} & \multicolumn{1}{c|}{\textbf{Innovation}} & \multicolumn{1}{c|}{\textbf{Shortcoming}} \\
\hline
\multirow{9}{*}{\textbf{\begin{tabular}[c]{@{}c@{}}Imperfect \\ synchronization \\ problem \end{tabular}}} & \multirow{5}{*}{\textbf{Radar}} & \cite{younis2006performance}  & \begin{tabular}[c]{@{}c@{}}Establish a dedicated synchronization link to obtain correction signal\\ and compensate for it to achieve high-accuracy phase\end{tabular}  & Cause data acquisition interruptions \\ \cline{3-5} 
                  &                   & \cite{jin2019advanced} & \begin{tabular}[c]{@{}c@{}}The synchronization signal is exchanged by the time difference\\ between the receiving and transmitting windows\end{tabular}  & Poor anti-noise performance   \\ \cline{3-5} 
                  &                   & \cite{cai2022advanced} & \begin{tabular}[c]{@{}c@{}}Exploit CS and MAP estimation to eliminate the noise\\ in synchronization phase\end{tabular}  & Unsuitable for communication signal \\ \cline{2-5} 
                  & \multirow{3}{*}{\textbf{WiFi}} &  \cite{qian2018widar2}  & \begin{tabular}[c]{@{}c@{}}Explore a CACC method to eliminate the synchronization error\\ in WiFi sensing system for human tracking \end{tabular} & Vulnerable to the multipaths \\ \cline{3-5} 
                  &                   &  
                  \cite{zeng2019farsense}  & Use the CASR method to obtain higher sensing accuracy than CACC  & \begin{tabular}[c]{@{}c@{}}Unsuitable for various sensing \\ information fusion \end{tabular}  \\ \hline
                  \multirow{9}{*}{\textbf{\begin{tabular}[c]{@{}c@{}}Sensing \\ information \\ fusion \end{tabular}}} & \multirow{4}{*}{\textbf{Data-level}} & \cite{weiss2009direct}  & \begin{tabular}[c]{@{}c@{}}A method of target location based on discrete grid point \\ and node parameter matching is proposed \end{tabular}  & Vulnerable to grid size and noise \\ \cline{3-5} 
                  &                   & \cite{ren2021improved}  & Propose a fusion method by filtering the abnormal and unstable values   & Discard some sensing information \\ \cline{3-5} 
                  &    & \cite{jiang2021research} & \begin{tabular}[c]{@{}c@{}} Explore a weighted data fusion method for multi-node \\ by utilizing the optimal weight assignment \end{tabular} & Inadequate use of sensing information  \\ \cline{2-5} 
                  & \multirow{3}{*}{\textbf{Symbol-level}} &  \cite{wei2023symbol}  &  \begin{tabular}[c]{@{}c@{}}Propose the concept of symbol-level fusion and  multi-BS \\ cooperative sensing signal processing method \end{tabular} & Unsuitable for multi-antenna system  \\ \cline{3-5} 
                  &                   &  \cite{liu2024target}  & \begin{tabular}[c]{@{}c@{}} Propose a symbol-level fusion method by using angle and delay \\information in macro and micro BSs cooperative sensing scenario \end{tabular} & Unsuitable for velocity estimation  \\ \hline   
\end{tabular} }
\end{table*}

\subsection{Related Work}
A summary of the related work about handling the imperfect synchronization and sensing information fusion problems is shown in Table \ref{tab_0}, and the details are as follows.

The imperfect synchronization problem in cooperative passive sensing has been studied in radar sensing and WiFi sensing.
In radar sensing, oscillator synchronization is usually 
aided by global positioning system (GPS). 
However, there are still synchronization errors in passive radar, which prevent further improvements in sensing accuracy. 
To this end, Younis \textit{et al.} in \cite{younis2006performance} 
proposed an alternate oscillator signal synchronization scheme that exchanges the oscillator signals by establishing a dedicated interstellar synchronization link to obtain the correction signal and compensates for it to achieve high-accuracy phase.
However, exchanging oscillator signals may cause
data acquisition interruptions.
Jin \textit{et al.} in \cite{jin2019advanced} 
adopted the time gap between the end of the echo reception window and the beginning of the next pulse repetition time to exchange the synchronization signals to improve the synchronization accuracy further.
However, the performance of this method is affected by noise. 
In response, Cai \textit{et al.} 
 in \cite{cai2022advanced} 
exploited compressed sensing (CS) and 
maximum a posteriori (MAP) estimation 
to eliminate the noise in synchronization 
to achieve high synchronization accuracy. 
The above methods are specific to radar signals, and it is unclear whether they are applicable to communication signals.
In terms of WiFi sensing, 
Zhang~\textit{et al.} summarized the methods to 
deal with clock desynchronization, 
including reference signal method, 
cross-antenna cross-correlation (CACC) method, 
and cross-antenna signal ratio (CASR) method \cite{zhang2022integration}.
The quality of the constructed reference signal significantly impacts the performance of the reference signal method.
Qian \textit{et al.} in \cite{qian2018widar2} used the CACC method to locate target 
with a single WiFi link for human tracking. 
However, this method is susceptible to the 
number of signal propagation paths. 
Therefore, Ni \textit{et al.} in \cite{ni2021uplink} proposed 
an alternative scheme called mirrored multiple signal classification (MUSIC), 
which exploits the symmetry between the unknown parameters and reduces the computational complexity.
However, the mirrored MUSIC method is limited by sensing coverage.
To this end, based on the CASR principle, 
Zeng \textit{et al.} in \cite{zeng2019farsense} used the 
channel state information (CSI) ratio of two received antennas to obtain better sensing coverage and accuracy than CACC.
The above CACC method amplifies noise
and CASR method cannot be used for the signal fusion of cooperative passive sensing.

The other challenge in cooperative passive sensing is 
sensing information fusion.
There are two main types of sensing information fusion methods in multi-BS cooperative passive sensing: 
data-level and symbol-level fusion.
Data-level fusion is to fuse the physical parameters of target estimated by multiple nodes, 
thereby estimating the location and velocity of target.
Weiss \textit{et al.} in \cite{weiss2009direct} 
explored a method of gridding the possible locations of target, combining the delay estimation parameter of each node to estimate the location of target.
However, the performance of this method is 
affected by grid size and noise.
To minimize the impact of noise on the sensing accuracy, 
Ren \textit{et al.} \cite{ren2021improved} presented 
a filtering and fusion method by evaluating the stability and consistency 
of sensing information of each node. 
However, this approach discards the sensing information of some nodes. 
In order to leverage the sensing information of all nodes, 
Jiang \textit{et al.} in \cite{jiang2021research} explored a weighted data fusion method for cooperative radar to 
obtain high-accuracy target position estimation 
by utilizing the optimal weight assignment.
Data-level fusion is simple to implement but has limited sensing accuracy.

Symbol-level fusion involves fusing the symbols of multiple nodes, 
where the phase information of each symbol contains the
target's parameters obtained from the sampled echo signal
after demodulation. This type of fusion, occurring prior to parameter estimation, can yield greater performance gains than data-level fusion~\cite{wei2024deep}.
Wei~\textit{et al.} in \cite{wei2023symbol} proposed the concept of symbol-level fusion and a symbol-level fusion method for multi-BS cooperative active sensing. However, \cite{wei2023symbol} does not take angle estimation into consideration.
To this end, Liu~\textit{et al.} in \cite{liu2024target} proposed a symbol-level fusion method using angle and delay parameters in macro and micro BSs cooperative sensing.

The symbol-level fusion improves the received signal-to-noise ratio (SNR) and 
obtains better sensing performance than data-level fusion. 
Therefore, this paper is dedicated to 
the symbol-level fusion approach.
However, the symbol-level fusion methods in the above literature do not investigate the Doppler frequency shift of target in passive sensing as well as fully utilizing the geometric relationship between target and multiple BSs.

\subsection{Our Contributions}
To address the above challenges,
we propose a high-accuracy symbol-level multi-BS 
cooperative passive sensing scheme applicable to 
mobile communication systems. 
The detailed contributions are as follows.
\begin{itemize}
    \item \textbf{Symbol-level multi-BS cooperative passive sensing:} 
    For target localization with symbol-level fusion of multi-BS sensing information, 
    the searching scope is gridded.
    Then, the feature vectors of target's position are 
    constructed and accumulated, 
    further being applied to search the 
    optimal estimation of target's location 
    from the grid points. 
    This method achieves noise suppression and 
    high-accuracy target localization.
    To estimate the absolute velocity of target, 
    the searching scopes of angles and magnitudes of target's velocity are gridded. Then, an expression for the Doppler frequency shift due to 
    the velocity of target is provided at the PBS.
    The velocity feature vectors are constructed and accumulated to 
    realize the high-accuracy absolute velocity estimation of target.
    \item \textbf{CFO and TO elimination method:} 
    For solving the imperfect synchronization problem between transmitting BS (TBS) and passive BS (PBS) 
    and avoiding ranging and velocity ambiguity, 
    we propose a non-line of sight (NLoS) and line of sight (LoS) signal cross-correlation (NLCC) method 
    on the PBS side. 
    The signals on the LoS and NLoS 
    paths between TBS and PBS are cross-correlated 
    to eliminate the synchronization error. 
    Simulation results demonstrate that this method mitigates 
    CFO and TO in asynchronous multiple BSs.
    \item \textbf{Joint AoA and AoD estimation method:}
    A low-complexity, high-accuracy joint AoA and AoD estimation method 
    is proposed to reduce the computational complexity of 
    the AoA and AoD estimation, which achieves 
    the same accuracy as the traditional two-dimensional MUSIC (2D-MUSIC) method 
    while reducing the computational complexity.
\end{itemize}

The rest of this paper is organized as follows. 
Section \ref{se2} presents the ISAC signal model. 
In Section \ref{se3}, the multi-BS sensing preprocessing is proposed. 
In Section \ref{se4}, the symbol-level fusion method of multi-BS sensing information is studied, 
including the estimation methods of location and absolute velocity. 
In Section \ref{se5}, the performance analysis of 
radar sensing is derived. 
Section \ref{se6} shows the simulation results and Section \ref{se7} summarizes this paper. 

\textit{\textbf{Notations:}} $\{\cdot\}$ typically stands for a set of various index values. 
Vectors and matrices are written as bold letters and capital bold letters, respectively. The $(m,n)$-th element of the matrix $\mathbf{H}$ is denoted by $\mathbf{H}|_{m,n}$. $\{H_i\}_{i=1}^I$ represents the set of $I$ elements.
$\mathbb{C}$ denotes the set of complex numbers. 
$\left(\cdot\right)^{\text{T}}$, $\left(\cdot\right)^{\text{H}}$, 
$\left(\cdot\right)^{\text{*}}$, and $\left(\cdot\right)^{-1}$ stand for the transpose operator, conjugate transpose operator, conjugate operator, and inverse operator, respectively. 
$\otimes$ and $\circ$ represent the Kronecker and Hadamard product, respectively.
$\arcsin\left({\cdot}\right)$ denotes the inverse sine function. $\left\lfloor\cdot\right\rfloor$ represents floor function. $\text{vec}(\cdot)$ represents the matrix vectorization.
A complex Gaussian random variable $\mathbf{u}$ with mean $\mu_u$ and variance $\sigma_u^2$ is denoted by $\mathbf{u} \sim \mathcal{CN}\left(\mu_u,\sigma_u^2\right)$.

\section{System Model} \label{se2}

\begin{figure}[!ht]
    \centering
    \includegraphics[width=0.45\textwidth]{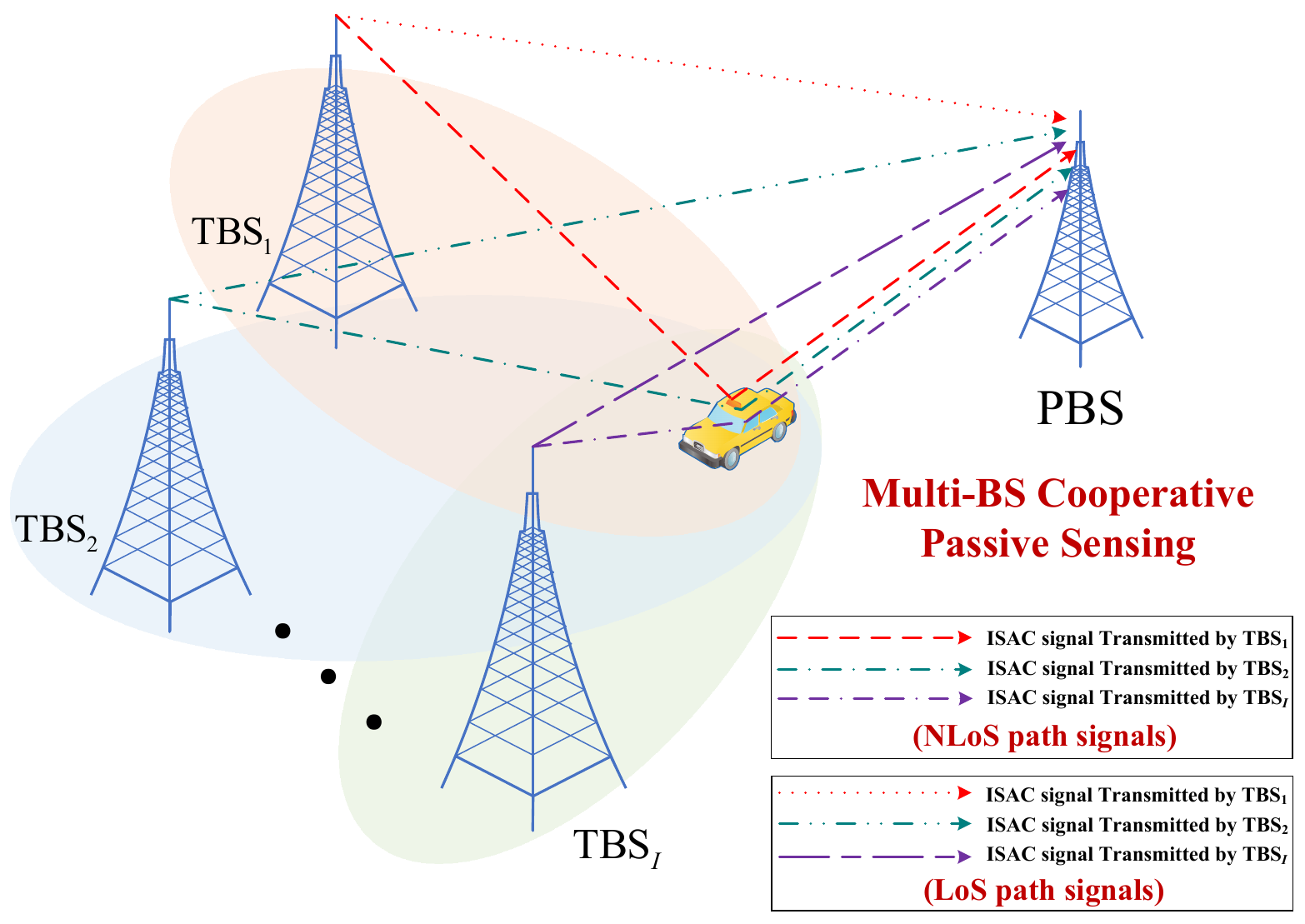}
    \caption{ISAC-enabled multi-BS cooperative passive sensing}
    \label{fig1:scenario}
\end{figure}

\subsection{System Setup}\label{sec2-A}

In this paper, we consider an ISAC-enabled multi-BS cooperative passive sensing scenario. 
As shown in Fig. \ref{fig1:scenario}, 
there are $I$ TBSs 
transmitting ISAC signals 
for passive sensing and downlink communication. 
PBS is receiving the downlink echo signal transmitted by TBSs 
from the LoS paths and the NLoS paths. 
The NLoS path refers to the path where the signal transmitted by TBS is reflected by target and reaches PBS. 
Each TBS employs time-division for communication and sensing, serving the users within its cell~\cite{liu2022integrated}.
Without loss of generality, 
there are two assumptions in this paper 
\cite{ni2021uplink,zhang2022integration,venkatesh2007non,wang2009new}.

\begin{itemize}
    \item 
    Since TBSs are static, 
    the PBS knows the locations of TBSs in advance.
    Then, the PBS can locate the absolute 
    location of target through relative distance estimation.
    \item 
    PBS knows the transmitted communication data of TBS in one of 
    the following two ways: 
    1) TBS and PBS are connected by fiber optic cables, 
    so that the PBS knows the transmitted data of TBS; 
    2) PBS could recognize and recover the communication data of TBS.
\end{itemize}

Meanwhile, each BS has a unique BS identity code
to distinguish the signals of different BSs. 
Therefore, the signals from different TBSs can be 
distinguished at the PBS side \cite{chouchane2009defending}.

\subsection{Signal Model}

In this paper, we consider the BS with
multiple input multiple output-OFDM (MIMO-OFDM) and uniform linear array (ULA).
For the $i\in \{1,2,\cdots,I\}$-th TBS, there are $N_\text{Tx}^{i}$ 
transmitting antennas and $N_\text{Rx}^{i}$ receiving antennas.
For each PBS, there are $N_\text{Rx}^{\text{p}}$ receiving antennas.

\subsubsection{Transmit signal model}
The transmit signals of $I$ ISAC TBSs adopt OFDM signals, which are generally defined as \cite{liu2024target,liu2024carrier,wei2023carrier} 
\begin{equation} \label{eq5}
x_{i}(t)=\sum_{m=1}^{M_{\text{sym}}}\sum_{n=1}^{N_{\text{c}}}d_{m,n}^ie^{j2\pi(f_\text{c}+n\Delta f)t}\text{rect}\left(\frac{t-(m-1)T}{T}\right),
\end{equation}
where $N_\text{c}$ and $M_\text{sym}$ denote the number of subcarriers and OFDM symbols, respectively; $n\in\left\{1,2,\cdots,N_\text{c}\right\}$ and $m\in\left\{1,2,\cdots,M_\text{sym}\right\}$ represent the indices of subcarriers and OFDM symbols, respectively; $d_{m,n}^i$ denotes the modulation symbol of the $i$-th TBS and $f_\text{c}$ is the carrier frequency; $\Delta f=1/T_\text{ofdm}$ is subcarrier spacing with $T_\text{ofdm}$ being the elementary OFDM symbol duration; $T=T_\text{ofdm}+T_\text{CP}$ is the total duration of OFDM symbol with $T_\text{CP}$ being the duration of cyclic prefix (CP); $\text{rect}\left(\cdot\right)$ represents a rectangular window function with width $T$.

\subsubsection{Received signal model}

According to Section \ref{sec2-A}, 
there are LoS and NLoS paths between TBSs and PBS 
in multi-BS cooperative passive sensing.
We assume that the real location of target is $\left(x_\text{tar},y_\text{tar}\right)$.
In terms of NLoS path, 
the distance between the target and the PBS is $r_\text{p,ns}$, 
the distance between the target and the $i$-th TBS is $r_{i,\text{ns}}$, 
and the corresponding delays are $\tau_\text{p,ns}=r_\text{p,ns}/c$ 
and $\tau_{i,\text{ns}}=r_{i,\text{ns}}/c$, respectively, where $c$ is the speed of light. 
In terms of LoS path, 
since the location of each BS is known, 
the distance between the $i$-th TBS and PBS is known 
and is set to $r_{i,\text{s}}$. 
Then, we have $\tau_{i,\text{s}}=r_{i,\text{s}}/c$.

The error of synchronization triggered by out-of-sync clocks and mismatched oscillators  between TBSs and PBS 
leads to potential time-varying TO and CFO \cite{wu2022joint}. 
Since this type of synchronization error originates from the hardware, we can assume that the electromagnetic signal (LoS and NLoS paths) between the TBSs and PBS carries the same synchronization error.
Therefore, in this paper, we use PBS as the reference point and define the TO between the $i$-th TBS and PBS as $\xi_{f,i}(m)$, and the CFO as $\xi_{\tau,i}(m)$. 
The phase shift of the signal in NLoS path 
consists of CFO, TO, delay $\tau_{\text{p},\text{ns}}^i=\tau_{i,\text{ns}}+\tau_{\text{p,ns}}$ 
and Doppler frequency shift $f_{i,\text{p}}$. 
$f_{i,\text{p}}$ is expressed as (\ref{eq9}) 
by \hyperref[theorem1]{\textbf{Theorem 1}}. 
The phase shift of the signal in LoS path 
includes path loss, delay, CFO, and TO.

\begin{theorem} \label{theorem1}
    The target is located between TBS and PBS 
    while traveling at an absolute velocity $\vec{v}$. 
    When the space coordinate systems of BS and target 
    are unified \cite{ISAC_jwj_2}, 
    the total Doppler frequency shift at PBS side is
    \begin{equation} \label{eq9}
       \begin{aligned}
            &f_{i,\text{p}} \\ & =\frac{-v f_\text{c}}c{\left[\cos\left(\theta-\phi_{\text{p},\text{ns}}^i\right)+\cos\left(\theta-\theta_{i,\text{ns}}\right)\right]}, \theta\in[0,2\pi), 
       \end{aligned}           
    \end{equation}
    where $v$ and $\theta$ are the magnitude and angle of target's absolute velocity, respectively; 
    $\theta_{i,\text{ns}}$ and $\phi_{\text{p,ns}}^i$ denote the AoD of 
the signal from the $i$-th TBS to target and the AoA of the signal 
reflected from target arriving at the PBS, respectively.
\end{theorem}

\begin{proof}
Assume that the angle of the velocity of target is along the counterclockwise rotation angle $\theta \in \left[0,2\pi\right)$, with a magnitude $v$. As shown in Fig. \ref{appendix1}, the target is located between TBS and PBS, and the angle of the target's velocity is arbitrary.
\begin{figure}[!ht]
    \centering
    \includegraphics[width=0.45\textwidth]{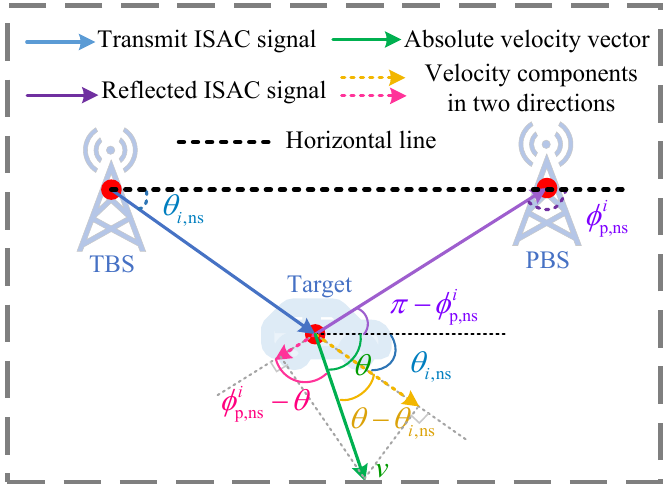}
    \caption{The locations of BSs and target}
    \label{appendix1}
\end{figure}

When the transmitted ISAC signal reaches the target, the movement of the target causes a Doppler frequency shift $f_{\text{p},1}^i=\frac{-v f_{\text{c}}\cos\left(\theta-\theta_{i,\text{ns}}\right)}c$ and the total carrier frequency is $f_{\text{p},2}^i=f_{\text{c}}+f_{\text{p},1}^i$. When the signal is reflected by the target and reaches the PBS, the echo ISAC signal generates a Doppler frequency shift $f_{\text{p},3}^i$, which is
\begin{equation} \label{apen_eq1}
    \begin{aligned}
       f_{\text{p},3}^i&=\frac{-v\cos\left(\theta-\phi_{\text{p},\text{ns}}^i\right)}{c}f_{\text{p},2}^i
       \\
       &=\frac{-v f_{\text{c}}\cos\left(\theta-\phi_{\text{p},\text{ns}}^i\right)}c\\
       &\quad +\frac{v^2f_{\text{c}}\cos\left(\theta-\phi_{\text{p},\text{ns}}^i\right)\cos\left(\theta-\theta_{i,\text{ns}}\right)}{c^2}.
    \end{aligned}
\end{equation}

It is worth noting that the carrier frequency is $f_{\text{p},2}^i$ instead of the original $f_\text{c}$. Therefore, the total Doppler frequency shift of the echo ISAC signal received at PBS with respect to the transmitted ISAC signal at TBS is
\begin{equation} \label{apen_eq2}
    \begin{aligned}  f_{i,\text{p}}&=f_{\text{p},1}^i+f_{\text{p},3}^i \\
   &=\frac{-v f_{\text{c}}\cos\left(\theta-\theta_{i,\text{ns}}\right)}c
   -\frac{v f_{\text{c}}\cos\left(\theta-\phi_{\text{p},\text{ns}}^i\right)}c\\
   &\quad +\frac{v^2f_{\text{c}}\cos\left(\theta-\phi_{\text{p},\text{ns}}^i\right)\cos\left(\theta-\theta_{i,\text{ns}}\right)}{c^2}.
    \end{aligned}
\end{equation}

As ${v^2}/{c^2}$ is very small, we ignore the quadratic term $\frac{v^2f_{\text{c}}\cos\left(\theta-\phi_{\text{p},\text{ns}}^i\right)\cos\left(\theta-\theta_{i,\text{ns}}\right)}{c^2}$ for simplicity, and (\ref{apen_eq2}) 
 is transformed to
\begin{equation} \label{apen_eq3}
   f_{i,\text{p}}=-\frac{v f_\text{c}}c{\left[\cos\left(\theta-\phi_{\text{p},\text{ns}}^i\right)+\cos\left(\theta-\theta_{i,\text{ns}}\right)\right]}.
\end{equation}
\end{proof}

\begin{figure*}[!ht]
    \begin{equation} \label{eq10}
        \begin{aligned}   \mathbf{y}_{\text{p}}^i\left(m,n\right)&=\mathbf{y}_{\text{p},\text{s}}^i\left(m,n\right)+\mathbf{y}_{\text{p},\text{ns}}^i\left(m,n\right)\\
       &=\underbrace{b_{\text{p},\text{s}}^i e ^ { j 2 \pi m T \xi _ { f , i }(m)}e^{-j2\pi n\Delta f\left[\tau_{i,\text{s}} + \xi _ { \tau , i }(m)\right]}\mathbf{a}_{\text{Rx},\text{p}}\left(\theta_{i,\text{s}}\right)\mathbf{a}_{\text{Tx}}^\text{T}\left(\theta_{i,\text{s}}\right)\mathbf{x}_{m,n}^i}_{\text{LoS path}} \\ & \quad + \underbrace{b_{\text{p},\text{ns}}^i e ^ { j 2 \pi m T \left [ f _ { i , \text{p}}+\xi_{f,i}\left(m\right)\right]}e^{-j2\pi n\Delta f\left[\tau_{\text{p},\text{ns}}^i+\xi_{\tau,i}\left(m\right)\right]}\mathbf{a}_{\text{Rx},\text{p}}\left(\phi_{\text{p},\text{ns}}^i\right)\mathbf{a}_{\text{Tx}}^{\text{T}}\left(\theta_{i,\text{ns}}\right)\mathbf{x}_{m,n}^i}_{\text{NLoS path}} + \mathbf{z}_\text{S}^{i} \in \mathbb{C}^{N_\text{Rx}^{\text{p}}\times 1}.
        \end{aligned}
    \end{equation}
    {\noindent} \rule[-10pt]{18cm}{0.1em}
\end{figure*}

According to \cite{sturm2011waveform}, 
for the PBS receiving the echo ISAC signal from the $i$-th TBS, 
the $i$-th baseband echo modulation symbol 
on the $n$-th subcarrier and the $m$-th OFDM symbol 
can be represented as (\ref{eq10}), 
where $b_{\text{p},\text{ns}}^i=\sqrt{\frac{\lambda^2}{\left(4\pi\right)^3r_{i,\text{ns}}^2r_{\text{p,ns}}^2}}\beta_{\text{p},\text{ns}}^i$ and $b_{\text{p},\text{s}}^i=\sqrt{\frac{\lambda^{2}}{\left(4\pi\right)^{3}r_{i,\text{s}}^{4}}}$ 
are the attenuations of LoS and NLoS paths, respectively.  
$\beta_{\text{p},\text{ns}}^i$ is the reflecting factor between the $i$-th TBS and target.
$\lambda=c/f_{\text{c}}$ is wavelength, $\mathbf{x}_{m,n}^i=\left[d_{m,n}^i,d_{m,n}^i,\cdots,d_{m,n}^i\right]\in\mathbb{C}^{N_{\text{Tx}}^{i}\times1}$ 
represents the data vector at the $m$-th OFDM symbol time. 
$\theta_{i,\text{s}}$ denotes the AoA and AoD of LoS path 
between the $i$-th TBS and PBS, which is known.
$\mathbf{z}_\text{S}^{i} \sim \mathcal{CN}\left(0,\sigma_\text{S}^{2}\right)$ 
is the AWGN of radar sensing channel. 
$\mathbf{a}_{\text{Rx},\text{p}}(\phi_{\text{p},\text{ns}}^i)$ 
and $\mathbf{a}_{\text{Tx}}\left(\theta_{i,\text{ns}}\right)$ 
denote the receive steering vector of PBS and 
the transmit steering vector of the $i$-th TBS, respectively, 
as shown in (\ref{eq11}).
\begin{equation} \label{eq11}
{\fontsize{9.5}{10}
    \begin{aligned}
     & \mathbf{a}_{\text{Rx},\text{p}}
    \left(\phi_{\text{p},\text{ns}}^i\right)=\\ & \left[e^{j2\pi(\frac{d_\text{r}}{\lambda}){\sin(\phi_{\text{p},\text{ns}}^i)}},e^{j4\pi(\frac{d_\text{r}}{\lambda}){\sin(\phi_{\text{p},\text{ns}}^i)}},\cdots,e^{j\left(N_\text{Rx}^\text{p}\right){2\pi(\frac{d_\text{r}}{\lambda}){\sin(\phi_{\text{p},\text{ns}}^i)}}}\right]^\text{T}, \\
    &\mathbf{a}_{\text{Tx}}\left(\theta_{i,\text{ns}}\right)= \\ & \left[e^{j{2\pi(\frac{d_\text{r}}{\lambda}){\sin(\theta_{i,\text{ns}})}}},e^{j{4\pi(\frac{d_\text{r}}{\lambda}){\sin(\theta_{i,\text{ns}})}}},\cdots,e^{j\left(N_\text{Tx}^{i}\right){2\pi(\frac{d_\text{r}}{\lambda}){\sin(\theta_{i,\text{ns}})}}}\right]^\text{T},  
    \end{aligned}  }
\end{equation}
where $d_\text{r}$ is the distance between two adjacent antennas.

There are $I$ TBSs in multi-BS cooperative sensing. 
Therefore, when the echo signals received by PBS come from $I$ TBSs, the echo modulation symbol received at PBS is
\begin{equation} \label{eq12}
\mathbf{y}_{\text{p}}=\sum_{i=1}^I
\sum_{m=1}^{M_{\text{sym}}}\sum_{n=1}^{N_{\text{c}}}\mathbf{y}_{\text{p}}^i(m,n).
\end{equation}

The location and absolute velocity of target are obtained 
by processing and fusing the echo modulation symbols received at PBS. 
Section \ref{se3} presents the preprocessing of the received data, 
including the estimation of AoA and AoD, 
the elimination of CFO and TO, 
and the coherent compression operation, etc. 
Section \ref{se4} investigates a symbol-level fusion method of multi-BS sensing information.

\section{Multi-BS Sensing Signal Preprocessing} \label{se3}

\begin{figure*}
    \centering
    \includegraphics[width=0.70\textwidth]{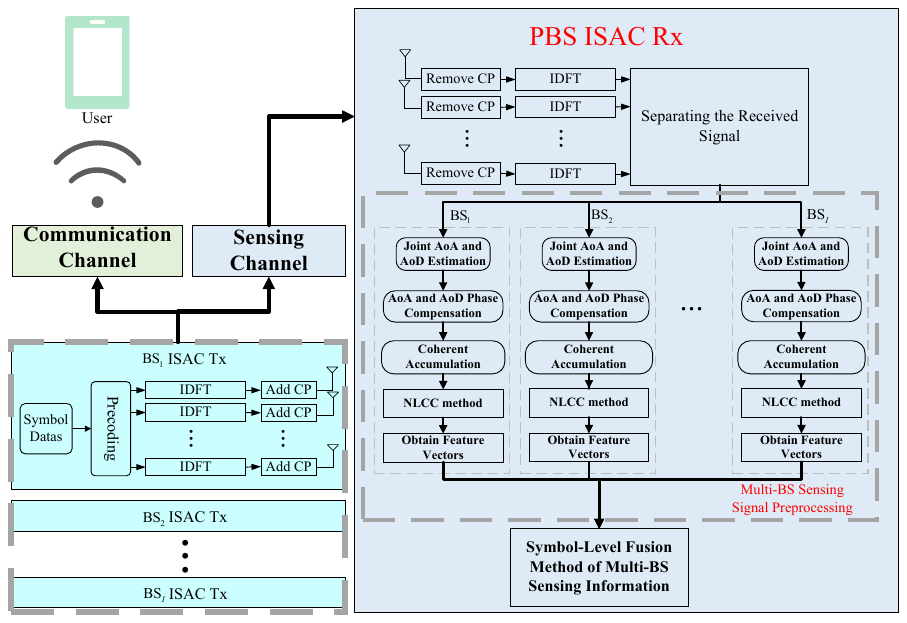}
    \caption{Multi-BS cooperative passive sensing scheme}
    \label{fig2}
\end{figure*}

As shown in Fig. \ref{fig2}, 
this section describes the preprocessing process 
of multi-BS sensing signals. 
Specifically, we separate the received signals by different BS identity code in PBS 
to obtain $I$ echo signals. 
For the $i$-th echo signal, 
we directly eliminate the AoA, AoD, 
and time delay of LoS path's signal via phase compensation, 
followed by coherent accumulation to obtain the delay-Doppler LoS information matrix 
$\mathbf{D}_{\text{p},\text{s}}^i$. 
For the NLoS path's signal, 
we initially estimate AoA and AoD by the proposed 
low-complexity joint AoA and AoD estimation method. 
Subsequently, phase compensation aligns 
the phase of time-frequency signals over all antennas. 
Coherent accumulation enhances the SNR, 
generating the delay-Doppler NLoS information matrix
$\mathbf{D}_{\text{p},\text{ns}}^i$.

CFO and TO in $\mathbf{D}_{\text{p},\text{ns}}^i$ 
will cause ambiguity in distance and velocity 
estimation~\cite{ni2021uplink}. 
To address this problem, 
a NLoS and LoS signal cross-correlation (NLCC) method is proposed 
for mitigating CFO and TO to generate the 
delay-Doppler information matrix $\mathbf{D}_{\text{p}}^i$. 
For $I$ delay-Doppler information matrices, 
we perform coherent compression processing 
to obtain range and velocity 
feature vectors (refer to Section \ref{se3-D} for details). 
Then, these feature vectors are utilized 
for symbol-level fusion of multi-BS sensing information in Section \ref{se4}.

\subsection{Low-complexity Joint AoA and AoD Estimation} \label{se3-1}

The traditional 2D-MUSIC method for joint AoA-AoD estimation 
in MIMO-OFDM ISAC system generally has 
high computational complexity 
\cite{pasya2014joint,zhang2010direction}. 
Therefore, a two-step low-complexity 
joint AoA and AoD estimation method is proposed in this section.

For the $i$-th echo NLoS signal 
(e.g., the NLoS path in (\ref{eq10})), 
when the transmitted data vector $\mathbf{x}_c^i$ 
is known, 
right-multiplying the generalized inverse matrix $\mathbf{X}_{m,n}^\dagger =\left(\mathbf{x}_{m,n}^{i}\right)^{\text{H}}\left(\mathbf{x}_{m,n}^{i}\left(\mathbf{x}_{m,n}^{i}\right)^{\text{H}}+\rho \mathbf{I}\right)^{-1}$ on $\mathbf{y}_{\text{p},\text{ns}}^i(m,n)$ yields

\begin{equation} \label{eq13}
   \begin{aligned}
&\mathbf{y}_{\text{p},\text{ns}}^{i}\big(m,n\big)\mathbf{X}_{m,n}^\dagger \\ & \approx b_{\text{p},\text{ns}}^{i} e ^ { j 2 \pi m T \left [ f _ { i , \text{p}}+\xi_{f,i}(m)\right]}e^{-j2\pi n\Delta f\left[\tau_{\text{p},\text{ns}}^{i}+\xi_{\tau,i}(m)\right]}  \\
&\quad\times\mathbf{a}_{\text{Rx},\text{p}}\left(\phi_{\text{p},\text{ns}}^i\right)\mathbf{a}_\text{Tx}^\text{T}\left(\theta_{i,\text{ns}}\right)\mathbf{I}+ \mathbf{z}_\text{ns}^i\mathbf{X}_{m,n}^\dagger,
    \end{aligned}
\end{equation}
where $\mathbf{I}\in\mathbb{C}^{N_{\text{Tx}}^{{i}}\times N_{\text{Tx}}^{{i}}}$ is identity matrix, $\mathbf{z}_\text{ns}^i$ is a AWGN vector. $\rho$ is a regularization parameter introduced to improve the accuracy and stability of the numerical recovery on the right side of (\ref{eq13})~\cite{tikhonov1963solution}. 
Therefore, the $(m,n)$-th all-antenna information matrix is denoted by
\begin{equation} \label{eq14}
\mathbf{Y}_{\text{p},\text{ns}}^i|_{m,n} = 
\mathbf{y}_{\text{p},\text{ns}}^{i}\big(m,n\big)\mathbf{X}_{m,n}^\dagger\in\mathbb{C}^{N_{\text{Rx}}^{\text{p}}\times N_{\text{Tx}}^{i}}.
\end{equation}

Observing (\ref{eq14}), 
it is discovered that the AoA introduces 
a linear phase shift along the received antenna elements of $\mathbf{Y}_{\text{p},\text{ns}}^i|_{m,n}$, 
while AoD introduces a linear phase shift 
along the transmit antenna elements of $\mathbf{Y}_{\text{p},\text{ns}}^i|_{m,n}$. 
The most critical observation is that 
the phase shifts introduced by AoA and AoD are 
completely orthogonal.

\subsubsection{Rough estimation stage} \label{se3-A-1}

In the rough estimation stage, since the initial OFDM symbol can avoid some kinds of interference, such as multipath interference and accumulated noise interference~\cite{hassan2017novel},
we combine the $(0,0)$-th all-antenna information matrix $\mathbf{Y}_{\text{p},\text{ns}}^i|_{0,0}$ with 2D-FFT method \cite{sturm2009ofdm,wei2023integrated} 
to obtain the rough estimation of AoA and AoD. 
Specifically, DFT is performed for 
each column and row of $\mathbf{Y}_{\text{p},\text{ns}}^i|_{0,0}$
to obtain rough estimation of AoA and AoD, respectively. 
The above operation is expressed 
as \cite{liu2022integrated}
\begin{equation} \label{eq15}    \mathbf{Y}=\mathbf{F}_{N_{\text{Rx}}^{\text{p}}}\mathbf{Y}_{\text{p},\text{ns}}^i|_{0,0}\mathbf{F}_{N_{\text{Tx}}^{i}}\in\mathbb{C}^{N_{\text{Rx}}^{\text{p}}\times{ N_{\text{Tx}}^{i}}},
\end{equation}
where $\mathbf{Y}$ denotes a 2D AoA-AoD profile 
and $\mathbf{F}_N \in \mathbb{C}^{N \times N}$ 
represents a DFT matrix. 
According to \hyperref[theorem2]{\textbf{Theorem 2}}, 
we obtain the estimation results
$\tilde{\phi}_{\text{p},\text{ns}}^i$ and 
$\tilde{\theta}_{i,\text{ns}}$, 
which narrow the searching range and reduce the computational
complexity of 2D-MUSIC method. 
The updated search intervals are utilized for 
the fine estimation stage.
\begin{theorem}
    \label{theorem2}
    When searching the peak value of $\mathbf{Y}$, 
    the corresponding indices for the AoA-axis and AoD-axis 
    are $\tilde{\mu}_\phi$ and $\tilde{\mu}_\theta$, respectively. 
    Therefore, the rough estimation results of AoA and AoD are
    \begin{equation} \label{eq16}
        \begin{cases}
        \begin{aligned}
        \tilde{\phi}_{\text{p},\text{ns}}^i &=
        \arcsin\left(\dfrac{\lambda\tilde{\mu}_\phi}{d_\text{r} N_{\text{Rx}}^{\text{p}}}\right),\\ \tilde{\theta}_{i,\text{ns}} &=\arcsin\left(\dfrac{\lambda\tilde{\mu}_\theta}{d_\text{r}N_{\text{Tx}}^{i}}\right). 
        \end{aligned}
        \end{cases}
    \end{equation}
    Additionally, the resolution of rough estimation is
    \begin{equation} \label{eq17}
        \begin{cases}
        \begin{aligned}
        \Delta\tilde{\phi}_{\text{p},\text{ns}}^i &=
        \arcsin\left(\dfrac{\lambda}{d_\text{r} N_{\text{Rx}}^{\text{p}}}\right),\\ \Delta\tilde{\theta}_{i,\text{ns}} &=\arcsin\left(\dfrac{\lambda}{d_\text{r}N_{\text{Tx}}^{i}}\right).
        \end{aligned}
        \end{cases}
    \end{equation}
\end{theorem}
\begin{proof}
To perform the DFT on the rows and columns of  $\mathbf{Y}_{\text{p},\text{ns}}^i|_{0,0}$, the first row of $\mathbf{Y}_{\text{p}}^i$ is denoted by $\mathbf{r}_{1}\big(k,\theta_{i,\text{ns}}\big)$ and the first column of $\mathbf{Y}_{\text{p},\text{ns}}^i|_{0,0}$ is denoted by $\mathbf{c}_{1}\big(J,\phi_{\text{p},\text{ns}}^{i}\big)$, where $k \in \{1,2,\cdots,N_{\text{Tx}}^{i}\}$ and $J \in \{1,2,\cdots,N_{\text{Rx}}^{\text{p}}\}$ are the indices of transmit antennas in the $i$-th TBS and  received antennas in PBS, respectively. Thus, the above operation is expressed as
\begin{equation} \label{apen_eq4}
    \begin{aligned}
&\mathbf{a}_{1}\left(\mu_{\phi}\right) =\text{DFT}\left[\mathbf{c}_1\left(J,\phi_{\text{p},\text{ns}}^i\right)\right]  \\
&=\sum_{J=1}^{N_\text{Rx}^\text{p}}\mathbf{c}_{1}\left(J,\phi_{\text{p},\text{ns}}^{i}\right)e^{-j\frac{2\pi}{N_\text{Rx}^\text{p}}J\mu_{\phi}} \\
&=\sum_{J=1}^{N_\text{Rx}^\text{p}}\left[b_{\text{p},\text{ns}}^{i}e^{j\frac{2\pi d_\text{r}}\lambda J \sin\left(\phi_{\text{p},\text{ns}}^i\right)}+\mathbf{z}_\text{ns}^i(J)\right]e^{-j\frac{2\pi}{N_\text{Rx}^\text{p}}J\mu_\phi},\\&\quad\quad \mu_\phi=1,2,\cdots,N_\text{Rx}^\text{p},
\end{aligned}
\end{equation}
and
\begin{equation} \label{apen_eq5}
    \begin{aligned}
&\mathbf{b}_{1}\left(\mu_{\theta}\right) =\text{DFT}\left[\mathbf{r}_{1}\left(k,\theta_{i,\text{ns}} \right ) \right ]  \\
&=\sum_{k=1}^{N_{\text{Tx}}^{i}}\mathbf{r}_1\left(k,\theta_{i,\text{ns}}\right)e^{-j\frac{2\pi}{N_{\text{Tx}}^{i}} k\mu_{\theta}} \\
&=\sum_{k=1}^{N_\text{Tx}^{i}}\left[b_{\text{p},\text{ns}}^{i}e^{jk\frac{2\pi d_\text{r}}\lambda\sin\left(\theta_{i,\text{ns}}\right)}+\mathbf{z}_\text{ns}^i(k)\right]e^{-j\frac{2\pi}{N_\text{Tx}^{i}}k{\mu_\theta}},\\&\quad\quad \mu_\theta=1,2,\cdots,N_\text{Tx}^{{i}},
\end{aligned}
\end{equation}
where $\mathbf{a}_{1}\left(\mu_{\phi}\right)$ and $\mathbf{b}_{1}\left(\mu_{\theta}\right)$ are AoA and AoD profiles, respectively. When searching for a peak, the corresponding AoA and AoD index values are
\begin{equation} \label{apen_eq6}
    \begin{aligned}
\tilde{\mu}_{\phi} &=\left\lfloor\dfrac{N_\text{Rx}^\text{p}d_\text{r}\sin\left(\tilde{\phi}_{\text{p},\text{ns}}^i\right)}\lambda\right\rfloor, \\
\tilde{\mu}_{\theta} &=\left\lfloor\frac{N_{\text{Tx}}^{i}d_{\text{r}}\sin\left(\tilde{\theta}_{i,\text{ns}}\right)}\lambda\right\rfloor.
    \end{aligned}
\end{equation}
Thus the rough estimations of AoA and AoD are obtained in (\ref{eq16}). By observing (\ref{apen_eq6}), the index values of  $\tilde{\mu}_{\phi}$ and $\tilde{\mu}_{\theta}$ only take integer values, whose resolution is $1$. The resolution of the corresponding rough estimation is shown in (\ref{eq17}).
\end{proof}

\subsubsection{Fine estimation stage}

In the rough estimation stage, 
the estimation results of $\tilde{\phi}_{\text{p},\text{ns}}^i$ 
and $\tilde{\theta}_{i,\text{ns}}$ are obtained. 
In this stage, the fine estimation is performed. 
Specifically, we utilize $\tilde{\phi}_{\text{p},\text{ns}}^i$ 
and $\tilde{\theta}_{i,\text{ns}}$ 
to update the searching range of 2D-MUSIC method, 
reducing computational complexity and obtaining 
high-accuracy estimation.

According to (\ref{apen_eq6}), 
only integers can be obtained in the rough estimation of 
$\tilde{\mu}_{\phi}$ and $\tilde{\mu}_{\theta}$. 
If the peak index values near the 
real AoA and AoD values are decimals, 
the real AoA and AoD values are around the estimated 
values in the rough estimation. 
Therefore, the searching range can be reduced to
\begin{equation} \label{eq18}   \begin{aligned}\phi_\text{search}^i=&\left[\arcsin\left(\frac{\lambda\left(\tilde{\mu}_\phi-1\right)}{d_\text{r}N_\text{Rx}^\text{p}}\right),\quad\arcsin\left(\frac{\lambda\left(\tilde{\mu}_\phi+1\right)}{d_\text{r}N_\text{Rx}^\text{p}}\right)\right],\\\theta_\text{search}^i=&\left[\arcsin\left(\frac{\lambda\left(\tilde{\mu}_\theta-1\right)}{d_\text{r}N_\text{Tx}^{i}}\right),\quad\arcsin\left(\frac{\lambda\left(\tilde{\mu}_\theta+1\right)}{d_\text{r}N_\text{Tx}^{i}}\right)\right].\end{aligned}
\end{equation}
For simplicity, (\ref{eq13}) can be expressed as
\begin{equation}  \label{eq19}
     \begin{aligned}
    &\mathbf{x}_{\text{p},\text{ns}}^i\left(m,n\right) \\ &=\mathbf{a}_{\text{Rx},\text{p}}\left(\phi_{\text{p},\text{ns}}^i\right)\otimes\mathbf{a}_{\text{Tx}}\left(\theta_{i,\text{ns}}\right)\alpha(m,n)+\mathbf{n}(m,n),
      \end{aligned}
\end{equation}
where $\alpha(m,n)$ denotes complex data 
and $\mathbf{n}(m,n)$ is AWGN. 
Then, the 2D-MUSIC method \cite{zhang2010direction} is 
applied with the updated searching ranges 
in fine estimation 
to obtain the final AoA estimation 
$\hat{\phi}_{\text{p},\text{ns}}^i$ and AoD estimation 
$\hat{\theta}_{i,\text{ns}}^i$. 
The 2D-MUSIC method is applied in this paper as follows.
\begin{itemize}
    \item \textbf{Step 1:}
    The $i$-th echo NLoS signal in (\ref{eq19}) is applied 
    to obtain the covariance matrix $\hat{\mathbf{R}}_{\mathbf{y}_{\text{p},\text{ns}}^i}$.
    \begin{equation} \label{eq20}
     {\fontsize{10}{10} \hat{\mathbf{R}}_{\mathbf{y}_{\text{p},\text{ns}}^i}  =\frac1{M_{\text{sym}}N_{\text{ c}}}\sum_{m=1}^{M_{\text{sym}}}\sum_{n=1}^{N_{\text{c}}}\mathbf{x}_{\text{p},\text{ns}}^{i}\left(m,n\right){\left[\mathbf{x}_{\text{p},\text{ns}}^{i}\left(m,n\right)\right]}^{\text{H}}.
      }
    \end{equation}
    
    \item \textbf{Step 2:} 
    The covariance matrix 
    $\hat{\mathbf{R}}_{\mathbf{y}_{\text{p},\text{ns}}^i}$ 
    undergoes an eigenvalue decomposition (EVD) to obtain 
    \cite{zhang2010direction}
    \begin{equation}   \label{eq21} \operatorname{eig}\left(\hat{\mathbf{R}}_{\mathbf{y}_{\text{p},\text{ns}}^i}\right)=\mathbf{E}_{s}\mathbf{D}_{s}\mathbf{E}_{s}^{\text{H}}+\mathbf{E}_{n}\mathbf{D}_{n}\mathbf{E}_{n}^{\text{H}},
    \end{equation}
    where $\mathbf{D}_s$ and $\mathbf{D}_n$ denote 
    the diagonal array of signal and noise, respectively, 
    while $\mathbf{E}_s$ and $\mathbf{E}_n$ denote the 
    sub-spaces of signal and noise, respectively.
    
    \item \textbf{Step 3:} 
    Construct the 2D-MUSIC spatial spectral function 
    $f_{\text{2d-music}}(\phi_{\text{s}},\theta_{\text{s}})$ 
    \cite{zhang2010direction}
    \begin{equation} \label{eq22}
        \begin{aligned}
            &f_{{\text{2d-music}} }(\phi_{\text{s}},\theta_{\text{s}}) \\ & =\frac1{\left[\mathbf{a}_{r}(\phi_{\text{s}})\otimes\mathbf{a}_{t}(\theta_{\text{s}})\right]^{\text{H}}\mathbf{E}_{n}\mathbf{E}_n^{\text{H}}\left[\mathbf{a}_{r}(\phi_{\text{s}})\otimes\mathbf{a}_{t}(\theta_{\text{s}})\right]},
        \end{aligned}
    \end{equation}
    where 
    \begin{equation} \label{eq23}
    { \fontsize{9}{10}
    \begin{gathered}
    \begin{aligned}
     &\mathbf{a}_{r}\left(\phi_{\text{s}}\right) = \\ &\left[\left.e^{j2\pi(\frac{d_{\text{r}}}{\lambda})\sin(\phi_{\text{s}})},e^{j4\pi(\frac{d_{\text{r}}}{\lambda})\sin(\phi_{\text{s}})},\cdots,e^{j\left(N_\text{Rx}^\text{p}\right)2\pi(\frac{d_{\text{r}}}{\lambda})\sin(\phi_{\text{s}})}\right]^\text{T},\right. \\
    & \mathbf{a}_{t}\left(\theta_{\text{s}}\right) = \\ &\left[e^{j2\pi(\frac{d_{\text{r}}}{\lambda})\sin(\theta_{\text{s}})},e^{4\pi(\frac{d_{\text{r}}}{\lambda})\sin(\theta_{\text{s}})},\cdots,e^{j\left(N_\text{Tx}^{i}\right)2\pi(\frac{d_{\text{r}}}{\lambda})\sin(\theta_{\text{s}})}\right]^\text{T}. 
     \end{aligned}
    \end{gathered}}
    \end{equation}
    $\phi_{\text{s}}$ and $\theta_{\text{s}}$ denote the 
    searching values of AoA and AoD, respectively.
 
    \item \textbf{Step 4:} 
     The peak of $f_{\text{2d-music}}(\phi_{\text{s}},\theta_{\text{s}})$ 
     is searched and the corresponding index values for 
     the axes are the final estimation results of AoA and AoD .
\end{itemize}

\hyperref[tab2]{\textbf{Algorithm 1}} 
provides a low-complexity joint AoA and AoD estimation method in Section \ref{se3-1}.

\begin{table}[!ht]
\centering
\label{tab2}
\resizebox{\linewidth}{!}{
\setlength{\arrayrulewidth}{1.5pt}
\begin{tabular}{rllll}
\hline
\multicolumn{5}{l}{\textbf{Algorithm 1:} Low-complexity joint AoA and AoD estimation}   \\ \hline
\multirow{-4}{*}{\textbf{Input:} }               & \multicolumn{4}{l}{\begin{tabular}[c]{@{}l@{}}All-antenna information matrix $\mathbf{Y}_{\text{p},\text{ns}}^i|_{0,0}$;\\ The number of OFDM symbols $M_{\text{sym}}$;\\ The number of subcarriers $N_{\text{c}}$;\\ The searching step size of 2D-MUSIC method $\varpi$. \end{tabular}} \\
\textbf{Output:}               & \multicolumn{4}{l}{The AoA and AoD estimation results $\hat{\phi}_{\text{p},\text{ns}}^i$ and $\hat{\theta}_{i,\text{ns}}$.} \\ 
\multicolumn{5}{l}{\textbf{Rough estimation stage:}} \\
\multirow{-1}{*}{1:}        & \multicolumn{4}{l}{Perform 2D-FFT on $\mathbf{Y}_{\text{p},\text{ns}}^i|_{0,0}$, as shown in (\ref{eq15}),} \\ & \multicolumn{4}{l}{the 2D AoA-AoD profile $\mathbf{Y}$ is obtained;}  \\
2:      & \multicolumn{4}{l}{$\textbf{for}$ ($\mu_\phi$ in $N_{\text{Rx}}^{\text{p}}$) and ($\mu_\theta$ in $N_{\text{Tx}}^{i}$) $\textbf{do}$}   \\
3:        & \multicolumn{4}{l}{$\hspace{1em}$ Get the maximum value of $\mathbf{Y}(\mu_\phi,\mu_\theta)$, denoted by $\mathbf{Y}(\tilde{\mu_\phi},\tilde{\mu_\theta})$;}  \\
4:        & \multicolumn{4}{l}{$\textbf{end}$ $\textbf{for}$}  \\
\multirow{-2}{*}{5:}         & \multicolumn{4}{l}{\begin{tabular}[c]{@{}l@{}} The AoA $\tilde{\phi}_{\text{p},\text{ns}}^i$ and AoD $\tilde{\theta}_{i,\text{ns}}$ are obtained via (\ref{eq16}) \\and the peak indices $\tilde{\mu_\phi}$ and $\tilde{\mu_\theta}$ of the 2D AoA-AoD profile; \end{tabular}}   \\
\multicolumn{5}{l}{\textbf{Fine estimation stage:}} \\
\multirow{-2}{*}{6:}     & \multicolumn{4}{l}{\begin{tabular}[c]{@{}l@{}}The updated searching ranges are obtained via (\ref{eq18}),\\ $\tilde{\mu_\phi}$, and $\tilde{\mu_\theta}$;\end{tabular}}  \\
7:                             & \multicolumn{4}{l}{The covariance matrix $\hat{\mathbf{R}}_{\mathbf{y}_{\text{p},\text{ns}}^i}$ is obtained by (\ref{eq20});}  \\
8:                             & \multicolumn{4}{l}{The noise sub-space $\mathbf{E}_n$ is obtained by (\ref{eq21});}  \\
9:                            & \multicolumn{4}{l}{$\textbf{for}$ $\phi_{\text{s}}, \theta_{\text{s}}$ in $\phi_{\text{search}}^i, \theta_{\text{search}}^i$ with step size $\varpi$ $\textbf{do}$}   \\
\multirow{-2}{*}{10:}                           & \multicolumn{4}{l}{\begin{tabular}[c]{@{}l@{}}$\hspace{1em}$ Spatial spectral functions $f_{\text{2d-music}}(\phi_{\text{s}},\theta_{\text{s}})$ are obtained \\ $\hspace{1em}$ via noise sub-space and (\ref{eq22});\end{tabular}}   \\
11:                            & \multicolumn{4}{l}{$\hspace{1em}$ The maximum value of $f_{\text{2d-music}}(\phi_{\text{s}},\theta_{\text{s}})$ is obtained;}   \\
12:                            & \multicolumn{4}{l}{$\hspace{1em}$  The corresponding AoA and AoD are obtained;}  \\
13:                            & \multicolumn{4}{l}{$\textbf{end}$ $\textbf{for}$ }    \\
\multirow{-2}{*}{14:}                            & \multicolumn{4}{l}{\begin{tabular}[c]{@{}l@{}}The final AoA and AoD in fine estimation are obtained, \\ denoted by $\hat{\phi}_{\text{p},\text{ns}}^i$ and $\hat{\theta}_{i,\text{ns}}$.\end{tabular}}  \\
\hline
\end{tabular}}
\end{table}

\subsection{Phase Compensation} \label{sec3-B}

In this subsection, 
the phase shifts in multiple antennas  
are compensated and 
coherent accumulation is performed to improve SNR. 
Specifically, 
the NLoS path signals and the LoS path signals
are processed.

\subsubsection{NLoS path signals}

When estimating the AoA $\hat{\phi}_{\text{p},\text{ns}}^i$ and the
AoD $\hat{\theta}_{i,\text{ns}}$, 
the phase shifts caused by AoA and AoD of 
the NLoS signal in (\ref{eq10}) can be eliminated.  
The NLoS part of the $i$-th baseband echo modulation symbol in the $J$-th receiving 
antenna on the $n$-th subcarrier and the $m$-th OFDM symbol time 
after compensation is
\begin{equation} \label{eq24}
    \begin{aligned}   p_{\text{p},\text{ns}}^i\left(J,m,n\right) & = 
 \kappa_{i,\text{ns}}b_{\text{p},\text{ns}}^{i}e^{j2\pi mT\left[f_{i,\text{p}}+\xi_{f,i}(m)\right]}\\ &\quad \times e^{-j2\pi n\Delta f\left[\tau_{\text{p},\text{ns}}^{i}+\xi_{\tau,i}(m)\right]}+z_\text{ns,offset}^i,
    \end{aligned}
\end{equation}
where $\kappa_{i,\text{ns}}$ denotes a real number generated by phase compensation operation, and $z_\text{ns,offset}^i$ is the noise.
According to (\ref{eq24}), 
the signals on multi-antenna do not change with $J$ and 
have the same phase, 
which can be coherently accumulated to improve SNR. 
Therefore, the NLoS part of the $i$-th baseband echo modulation symbol on 
the $n$-th subcarrier and the $m$-th OFDM symbol 
after coherent accumulation can be expressed as
\begin{equation} \label{eq25}
d_{\text{p},\text{ns}}^i\left(m,n\right)=\frac1{N_\text{Rx}^\text{p}}\sum_{J=1}^{N_\text{Rx}^\text{p}}p_{\text{p},\text{ns}}^i\left(J,m,n\right).
\end{equation}

The matrix form of (\ref{eq25}) on 
$N_\text{c}$ subcarriers and $M_\text{sym}$ OFDM symbols is shown in (\ref{eq26}). 
$\mathbf{D}_{\text{p},\text{ns}}^{i}\in\mathbb{C}^{N_{\text{c}}\times M_{\text{sym}}}$ 
is referred to as delay-Doppler NLoS information matrix, and $\mathbf{Z}_\text{ns}^i$ is the noise matrix.

\begin{figure*} 
   \begin{equation}
    \centering
    \label{eq26}
    \mathbf{D}_{\text{p},\text{ns}}^i=\left[
    \begin{array}{ccccc}
    d_{\text{p},\text{ns}}^i\left(1,1\right) & \cdots & d_{\text{p},\text{ns}}^i\left(1,m\right) & \cdots & d_{\text{p},\text{ns}}^i\left(1,M_{\text{sym}}\right) \\
     \vdots & \ddots & \vdots & \ddots & \vdots \\
    d_{\text{p},\text{ns}}^i\left(n,1\right)& \cdots & d_{\text{p},\text{ns}}^i\left(n,m\right) & \cdots & d_{\text{p},\text{ns}}^i\left(n,M_{\text{sym}}\right) \\
     \vdots & \ddots & \vdots & \ddots & \vdots \\
    d_{\text{p},\text{ns}}^i\left(N_\text{c},1\right)& \cdots & d_{\text{p},\text{ns}}^i\left(N_\text{c},m\right) & \cdots & d_{\text{p},\text{ns}}^i\left(N_\text{c},M_{\text{sym}}\right)\\
    \end{array} \right]+\mathbf{Z}_\text{ns}^i.
    \end{equation} 
\end{figure*}

\subsubsection{LoS path signals}
In the LoS path signals, 
AoA, AoD, and delay are known and can be directly eliminated, enabling coherent accumulation of the LoS path signals after removing these known parameters, including $\theta_{i,\text{s}}$ and $\tau_{i,\text{s}}$.
Therefore, the LoS part of the $i$-th baseband echo modulation symbol 
on the $n$-th subcarrier and the $m$-th OFDM symbol 
after coherent accumulation is expressed as
\begin{equation} \label{eq27}
    \begin{aligned} &p_{\text{p},\text{s}}^i\left(J,m,n\right) \\ &=\kappa_{i,\text{s}}b_{\text{p},\text{s}}^{i}e^{j2\pi mT\xi_{f,i}(m)}e^{-j2\pi n\Delta f\xi_{\tau,i}(m)}+z_\text{s,offset}^i,
    \end{aligned}
\end{equation}
where $\kappa_{i,\text{s}}$ denotes a real number generated by phase compensation operation, and $z_\text{s,offset}^i$ is the noise. Similar to NLoS path signals, 
the result of the coherent accumulation of multi-antenna data 
in LoS path is referred to as the delay-Doppler LoS information matrix 
$\mathbf{D}_{\text{p},\text{s}}^i \in \mathbb{C}^{N_{\text{c}}\times M_{\text{sym}}}$.

\subsection{NLCC method} \label{se3-C}

When $\mathbf{D}_{\text{p},\text{ns}}^{i}$ is utilized for 
the estimation of location and velocity of target, 
the CFO and TO lead to the ambiguity \cite{ni2021uplink}. 
Therefore, NLCC method is proposed to mitigate the CFO and TO.

Observing $\mathbf{D}_{\text{p},\text{ns}}^i$ and 
$\mathbf{D}_{\text{p},\text{s}}^i$ obtained in Section \ref{sec3-B}, 
the NLCC method is performed to mitigate the CFO and TO 
since both $\mathbf{D}_{\text{p},\text{ns}}^i$ and 
$\mathbf{D}_{\text{p},\text{s}}^i$ contain the same CFO and TO. 
With the cross-correlation on the corresponding elements of $\mathbf{D}_{\text{p},\text{ns}}^i$ and $\mathbf{D}_{\text{p},\text{s}}^i$, 
the delay-Doppler information matrix $\mathbf{D}_\text{p}^i \in \mathbb{C}^{N_{\text{c}}\times M_{\text{sym}}}$ is obtained, 
as shown in (\ref{eq28}), 
where $\mathbf{N}_\text{S}^i$ is the noise matrix.
\begin{figure*}[ht]
    \begin{equation} \label{eq28}
        \begin{aligned}
\mathbf{D}_{\text{p}}^{i}&=\mathbf{D}_{\text{p},\text{ns}}^{i}\circ\left(\mathbf{D}_{\text{p},\text{s}}^{i}\right)^{*} \\
&=\kappa_{i,\text{s}}\kappa_{i,\text{ns}}b_{\text{p},\text{s}}^{i}b_{\text{p},\text{ns}}^{i} \\
&\quad \times\left[ 
       \begin{array}{ccccc}
        e^{j2\pi Tf_{i,\text{p}}}e^{-j2\pi\Delta f\tau_{\text{p},\text{ns}}^{i}} & \cdots & e^{j2\pi mTf_{i,\text{p}}}e^{-j2\pi\Delta f\tau_{\text{p},\text{ns}}^{i}}  & \cdots & e^{j2\pi M_{\text{sym}} Tf_{i,\text{p}}}e^{-j2\pi\Delta f\tau_{\text{p},\text{ns}}^{i}} \\
        \vdots & \ddots & \vdots & \ddots & \vdots \\
        e^{j2\pi Tf_{i,\text{p}}}e^{-j2\pi n \Delta f\tau_{\text{p},\text{ns}}^{i}} &  \cdots & e^{j2\pi mTf_{i,\text{p}}}e^{-j2\pi n \Delta f\tau_{\text{p},\text{ns}}^{i}} & \cdots & e^{j2\pi M_{\text{sym}} Tf_{i,\text{p}}}e^{-j2\pi n \Delta f\tau_{\text{p},\text{ns}}^{i}} \\
        \vdots & \ddots & \vdots & \ddots & \vdots \\
        e^{j2\pi Tf_{i,\text{p}}}e^{-j2\pi N_{\text{c}} \Delta f\tau_{\text{p},\text{ns}}^{i}} & \cdots & e^{j2\pi mTf_{i,\text{p}}}e^{-j2\pi N_{\text{c}} \Delta f\tau_{\text{p},\text{ns}}^{i}} & \cdots & e^{j2\pi M_{\text{sym}} Tf_{i,\text{p}}}e^{-j2\pi N_{\text{c}} \Delta f\tau_{\text{p},\text{ns}}^{i}}
       \end{array}\right]+\mathbf{N}_\text{S}^i.
       \end{aligned}
    \end{equation}
    {\noindent} \rule[-10pt]{18cm}{0.1em}
\end{figure*}

When CFO and TO are eliminated, 
the delay-Doppler information matrices undergo a coherent compression 
operation to reduce the data processing load while improving SNR. 
Section~\ref{se3-D} describes the coherent compression operation in detail.

\subsection{Coherent Compression Operation}\label{se3-D}

Observing $\mathbf{D}_{\text{p}}^i$, 
the initial phase of each row vector or each column vector is not aligned.
The phase offset caused by delay remains constant across the elements within a subcarrier, 
while the Doppler effect maintains uniform phase shifts across the elements within an OFDM symbol.
Therefore, a coherent compression operation is proposed for processing $\mathbf{D}_{\text{p}}^i$ to obtain the range feature vectors 
and velocity feature vectors, which can reduce the data processing load while improving SNR.
Coherent compression refers to the inner product of the row or column vectors of a matrix, thereby transforming the matrix into a vector. The feature vector is a vector where the phase shift (delay or Doppler) caused by the target parameter varies linearly along the vector.
The operational process unfolds as follows.

\subsubsection{Range feature vector}

For $\mathbf{D}_{\text{p}}^i$, 
the $n$-th row vector of $\mathbf{D}_{\text{p}}^i$ 
is denoted by $\mathbf{f}_{\text{p}}^{i,n}$.
To align the phase of each row vector, 
$\mathbf{f}_{\text{p}}^{i,1}$ is chosen as 
the reference row vector 
and the remaining $N_{\text{c}}-1$ row vectors are conjugated with 
$\mathbf{f}_{\text{p}}^{i,1}$ to obtain $N_{\text{c}}-1$ 
coherent row vectors. 
The $n' \in \left\{ 1,2,\cdots,N_{\text{c}}-1\right\}$-th 
coherent row vector is expressed as
\begin{equation} \label{eq29}
    \begin{aligned}
&\mathbf{f}_{\text{p},1}^{i,n'} \in\mathbb{C}^{1\times M_{\text{sym}}} =\mathbf{f}_{\text{p}}^{i,n'+1}\circ\left(\mathbf{f}_{\text{p}}^{i,1}\right)^*  
=\\ &\left[e^{-j2\pi n'\Delta f\tau_{\text{p},\text{ns}}^{i}},e^{-j2\pi n'\Delta f\tau_{\text{p},\text{ns}}^{i}},\cdots,e^{-j2\pi n'\Delta f\tau_{\text{p},\text{ns}}^{i}}\right]+\mathbf{n}_\text{S}^i,
\end{aligned}
\end{equation}
where $\mathbf{n}_\text{S}^i$ denotes the noise vector.

It is evident that by summing the elements in $\mathbf{f}_{\text{p},1}^{i,n'}$, SNR can be improved, as the elements are identical, while the accompanying noise remains random.
Therefore, the elements in each of the $N_\text{c}-1$ 
coherent row vectors are summed to obtain the range feature vector 
$\mathbf{f}_{\text{p}}^i$, as shown in (\ref{eq30}).

\begin{figure*}
    \begin{equation}  \label{eq30}     \mathbf{f}_{\text{p}}^i=\frac1{M_{\text{sym}}}\left[\sum_{m=1}^{M_{\text{sym}}}\mathbf{f}_{\text{p},1}^{i,1}(m),\cdots, \sum_{m=1}^{M_{\text{sym}}}\mathbf{f}_{\text{p},1}^{i,n'}(m),\cdots,\sum_{m=1}^{M_{\text{sym}}}\mathbf{f}_{\text{p},1}^{i,N_{\text{c}}-1}(m)\right]^{\text{T}}\in\mathbb{C}^{(N_{\text{c}}-1)\times1}.
    \end{equation}
\end{figure*}

\subsubsection{Velocity feature vector}
Similar to the processing method of range feature vectors, 
the $m$-th column vector of $\mathbf{D}_{\text{p}}^i$ 
is denoted by $\mathbf{e}_{\text{p}}^{i,m}$. 
The $\mathbf{e}_{\text{p}}^{i,1}$ is chosen 
as the reference column vector and 
the remaining $M_{\text{sym}}-1$ column vectors are conjugated 
with $\mathbf{e}_{\text{p}}^{i,1}$ to obtain $M_{\text{sym}}-1$ 
coherent column vectors. 
The $m' \in \left\{ 1,2,\cdots,M_{\text{sym}}-1\right\}$-th 
coherent column vector is 
\begin{equation} \label{eq31}
    \begin{aligned}
&\mathbf{e}_{\text{p},1}^{i,m'} \in\mathbb{C}^{N_{\text{c}}\times1 } =\mathbf{e}_{\text{p}}^{i,m'+1}\circ\left(\mathbf{e}_{\text{p}}^{i,1}\right)^*  \\
&=\left[e^{j2\pi m'Tf_{i,\text{p}}},e^{j2\pi m'Tf_{i,\text{p}}}, \cdots,e^{j2\pi m'Tf_{i,\text{p}}}\right]^\text{T}+\mathbf{m}_\text{S}^i,
\end{aligned}
\end{equation}
where $\mathbf{m}_\text{S}^i$ is the noise vector.

Similarly, the elements of $\mathbf{e}_{\text{p},1}^{i,m'}$ 
are identical. 
Therefore, we perform an operation similar to 
$\mathbf{f}_{\text{p},1}^{i,n'}$ to obtain the velocity feature 
vector $\mathbf{e}_{\text{p}}^{i}$, 
as shown in (\ref{eq32}).

\begin{figure*}
    \begin{equation}  \label{eq32}     \mathbf{e}_{\text{p}}^i =\frac1{N_{\text{c}}}\left[\sum_{n=1}^{N_{\text{c}}}\mathbf{e}_{\text{p},1}^{i,1}(n),\cdots, \sum_{n=1}^{N_{\text{c}}}\mathbf{e}_{\text{p},1}^{i,m'}(n),\cdots,\sum_{n=1}^{N_{\text{c}}}\mathbf{e}_{\text{p},1}^{i,M_{\text{sym}}-1}(n)\right]\in\mathbb{C}^{1\times (M_{\text{sym}}-1)}.
    \end{equation}
    {\noindent} \rule[-10pt]{18cm}{0.1em}
\end{figure*}

The $I$ signals from TBSs are processed in Section \ref{se3} 
to obtain $I$ range feature vectors and $I$ velocity feature vectors, 
respectively. 
The above feature vectors are used to 
perform symbol-level fusion of multi-BS sensing information.

\section{Symbol-Level Fusion of Multi-BS Sensing Information} 
\label{se4}
In this section,
a symbol-level fusion method of multi-BS sensing information 
is proposed 
for cooperative passive sensing, 
including the location and absolute velocity estimation of target. 
Specifically, in the location estimation, 
the searching scope is first established.
Given that the target is located in the overlapping coverage of multiple TBSs, whose location and coverage are known, the searching scope is calculated. Then, the searching scope is gridded to construct the position searching scope matrix (in (\ref{eq33})). Utilizing the position searching scope matrix, the distance vector (in (\ref{eq35})) associated with the $i$-th TBS is obtained. Subsequently, the distance compensation operation is performed, whose value depends on the distance between the searching location and the PBS. The compensated distance vector (in (\ref{eq37})) is transformed into a delay matching matrix (in (\ref{eq38})) and is performed conjugate inner product with the corresponding $i$-th range feature vector, resulting in the $i$-th position profile (in (\ref{eq39})). The $I$ position profiles are coherently accumulated, and the index value corresponding to the peak value is the final location estimation of target. A similar process is carried out for absolute velocity estimation. Fig.~\ref{fig3} shows the process of symbol-level fusion method of multi-BS sensing information. A detailed description of the localization method and velocity estimation method is presented as follows.
\begin{figure}
    \centering
    \includegraphics[width=0.5\textwidth]{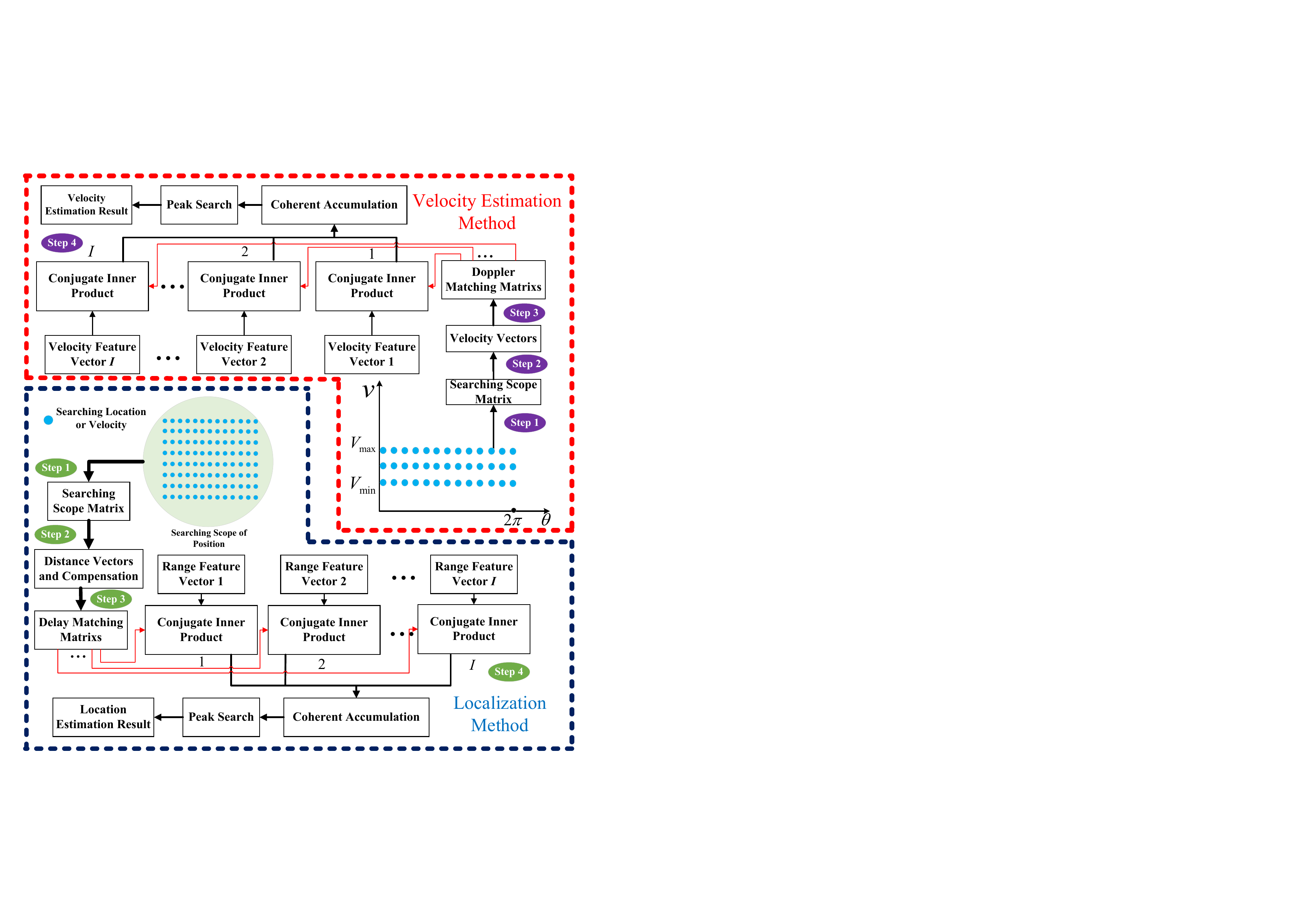}
    \caption{Flowchart of symbol-level fusion method of multi-BS sensing information}
    \label{fig3}
\end{figure}

\subsection{Localization Method}
The localization method is given in detail in this section. In Section \ref{se3}, we obtain $I$ range and velocity feature vectors. Meanwhile, the searching scope of target is obtained. Then, the searching scope is gridded to get the position searching scope matrix.
\subsubsection{Obtain position searching scope matrix} \label{se4-1-1}
The searching scope of length and width $Q\Delta M$ is gridded with a step size $\Delta M$ to construct a position searching scope matrix $\mathbf{P}_{\text{search}} \in \mathbb{C}^{Q \times Q}$.
\begin{equation} \label{eq33}
    \begin{aligned}
    &\mathbf{P}_{\text{search}}= \\ &\left[
    \begin{array}{cccc}
     \left(x_{\text{ra},1},y_{\text{ra},1}\right)  & \left(x_{\text{ra},1},y_{\text{ra},2}\right)  & \cdots & \left(x_{\text{ra},1},y_{\text{ra},Q}\right)  \\
     \left(x_{\text{ra},2},y_{\text{ra},1}\right)  & \left(x_{\text{ra},2},y_{\text{ra},2}\right)  & \cdots & \left(x_{\text{ra},2},y_{\text{ra},Q}\right)  \\
     \vdots & \vdots & \ddots & \vdots \\
     \left(x_{\text{ra},Q},y_{\text{ra},1}\right)  & \left(x_{\text{ra},Q},y_{\text{ra},2}\right)  & \cdots & \left(x_{\text{ra},Q},y_{\text{ra},Q}\right)  \\
    \end{array}
    \right],
    \end{aligned}
\end{equation}
where $\left(x_{\text{ra},p},y_{\text{ra},j}\right)$ indicates the coordinates within the searching scope with $p,j \in \left\{1,2,\cdots,Q\right\}$.

\subsubsection{Obtain distance vectors and compensation} \label{se4-1-2}
We assume that the coordinates of the $i$-th TBS are $\left(x^i,y^i\right)$, while transforming $\mathbf{P}_\text{search}$ into row vector form to obtain the distance searching vector $\mathbf{p}_{\text{search}}\in \mathbb{C}^{1\times Q^2}$.
\begin{equation} \label{eq34}
    \mathbf{p}_{\text{search}}=
    \left[\text{vec}(\mathbf{P}_{\text{search}})
    \right]^{\text{T}}.
\end{equation}
Then, the distance between each coordinate in $\mathbf{p}_{\text{search}}$ and the location of the $i$-th TBS is calculated to obtain the $i$-th distance vector $\mathbf{g}_{\text{r}}^i \in \mathbb{C}^{1 \times Q^2}$.
\begin{equation} \label{eq35}
    \mathbf{g}_{\text{r}}^i =\left[
    r_{i,1},r_{i,2},\cdots,r_{i,Q^2}\right],
\end{equation}
where $r_{i,1}=\sqrt{\left(x^i-x_{\text{ra},1}\right)^2+\left(y^i-y_{\text{ra},1}\right)^2}$.
However, $\mathbf{g}_{\text{r}}^i$ contains only the distance from the $i$-th TBS to the searching location, which needs to be compensated for the distance from the searching location to the PBS. Therefore, we calculate the distance between each coordinate in $\mathbf{p}_{\text{search}}$ and the PBS to obtain the distance compensation vector $\mathbf{g}_{\text{cop}} \in \mathbb{C}^{1 \times Q^2}$.
\begin{equation} \label{eq36}
    \mathbf{g}_{\text{cop}}=\left[
    r_{\text{p},1},r_{\text{p},2},\cdots,r_{\text{p},Q^2}\right].
\end{equation}
Then, the $I$ vectors $\mathbf{g}_{\text{r}}^i$ are summed with $\mathbf{g}_{\text{cop}}$ to get $I$ compensated distance vectors, where the $i$-th compensated distance vector is
\begin{equation} \label{eq37}    \mathbf{g}_{\text{r},\text{cop}}^{i}=\left[r_{\text{p},1}^{i},r_{\text{p},2}^{i},\cdots,r_{\text{p},Q^2}^{i}\right]\in \mathbb{C}^{1 \times Q^2},
\end{equation}
where $r_{{\text{p},p_\text{d}}}^{i}=r_{i,p_{\text{d}}}+r_{\text{p},p_{\text{d}}}$, $p_{\text{d}} \in \left\{1,2,\cdots,Q^2\right\}$.

\subsubsection{Obtain distance matching matrices} \label{se4-1-3}
According to the form of phase shifts generated by the delay along the frequency axis in \cite{sturm2011waveform}, we transform the $I$ compensated distance vectors into delay matching matrices, where the $i$-th delay matching matrix $\mathbf{G}_{\text{mat}}^i \in \mathbb{C}^{(N_{\text{c}}-1) \times Q^2}$ is shown in (\ref{eq38}). 
\begin{figure*}
    \begin{equation} \label{eq38}
        \mathbf{G}_{\text{mat}}^i= \left[
         \begin{array}{ccccc}
           e^{-j2\pi \Delta f \left(\dfrac{r_{\text{p},1}^{i}}{c}\right)}  & \cdots & e^{-j2\pi \Delta f \left(\dfrac{r_{\text{p},p_{\text{d}}}^{i}}{c}\right)} & \cdots &
           e^{-j2\pi \Delta f \left(\dfrac{r_{\text{p},Q^2}^{i}}{c}\right)} \\
           \vdots & \ddots & \vdots & \ddots & \vdots \\
           e^{-j2\pi n' \Delta f \left(\dfrac{r_{\text{p},1}^{i}}{c}\right)}  & \cdots & e^{-j2\pi n' \Delta f \left(\dfrac{r_{\text{p},p_{\text{d}}}^{i}}{c}\right)} & \cdots &
           e^{-j2\pi n' \Delta f \left(\dfrac{r_{\text{p},Q^2}^{i}}{c}\right)} \\
           \vdots & \ddots & \vdots & \ddots & \vdots \\
           e^{-j2\pi \left(N_{\text{c}}-1\right) \Delta f \left(\dfrac{r_{\text{p},1}^{i}}{c}\right)}  & \cdots & e^{-j2\pi \left(N_{\text{c}}-1\right) \Delta f \left(\dfrac{r_{\text{p},p_{\text{d}}}^{i}}{c}\right)} & \cdots &
           e^{-j2\pi \left(N_{\text{c}}-1\right) \Delta f \left(\dfrac{r_{\text{p},Q^2}^{i}}{c}\right)} \\
         \end{array}         
         \right].       
    \end{equation}
\end{figure*}

\subsubsection{Conjugate inner product}
After processing in Sections \ref{se4-1-1}, \ref{se4-1-2}, and \ref{se4-1-3}, we obtain $I$ delay matching matrices in (\ref{eq38}). In Section \ref{se3}, we obtain $I$ range feature vectors in (\ref{eq30}). In the following, we utilize $\{\mathbf{G}_{\text{mat}}^i\}_{i=1}^I$ and $\{\mathbf{f}_{\text{p}}^i\}_{i=1}^I$ to estimate the location of target.

The $i$-th $\mathbf{G}_{\text{mat}}^i$ and the corresponding $\mathbf{f}_{\text{p}}^i$ are performed conjugate inner product operation, which can be replaced by matrix multiplication. The $i$-th position profile $\mathbf{p}_i$ is shown in (\ref{eq39}), where $\tau_{\text{p},p_\text{d}}^{i} = r_{\text{p},p_\text{d}}^{i}/c$. Then, a coherent accumulation operation is carried out to $I$ position profiles to obtain the final position profile $\mathbf{p}_{\text{fn}}$.
\begin{figure*}
    \begin{equation} \label{eq39}
        \begin{aligned}
&\mathbf{p}_i=\left(\mathbf{f}_{\text{p}}^i\right)^{\text{H}}\mathbf{G}_{\text{mat}}^i \\
&=\left[\sum_{n'=1}^{N_{\text{c}}-1}e^{-j2\pi n' \Delta f\left(\tau_{\text{p},1}^{i}-\tau_{\text{p},\text{ns}}^{i}\right)},\sum_{n'=1}^{N_{\text{c}}-1}e^{-j2\pi n' \Delta f\left(\tau_{\text{p},2}^{i}-\tau_{\text{p},\text{ns}}^{i}\right)},\cdots,\sum_{n'=1}^{N_{\text{c}}-1}e^{-j2\pi n' \Delta f\left(\tau_{\text{p},Q^2}^{i}-\tau_{\text{p},\text{ns}}^{i}\right)}\right]\in\mathbb{C}^{1\times Q^2}.
\end{aligned}
    \end{equation}
     {\noindent} \rule[-10pt]{18cm}{0.1em} 
\end{figure*}
\begin{equation}\label{eq40}
    \mathbf{p}_\text{fn}=\frac1I\sum_{i=1}^I\mathbf{p}_i,
\end{equation}
the peak value of which is searched to obtain the peak index value $\tilde{p_{\text{d}}}$. The value of $\mathbf{p}_{\text{search}}(\tilde{p_{\text{d}}})$ is the final estimated location of target, denoted by $(\hat{x}_{\text{tar}},\hat{y}_{\text{tar}})$.

\hyperref[tab3]{\textbf{Algorithm 2}} demonstrates the localization method in the symbol-level fusion method of multi-BS sensing information.
\begin{table}[!ht]
\centering
\label{tab3}
\resizebox{0.48\textwidth}{!}{
\setlength{\arrayrulewidth}{1.5pt}
\begin{tabular}{rllll}
\hline
\multicolumn{5}{l}{\textbf{Algorithm 2:} Localization method}   \\ \hline
\multirow{-5}{*}{\textbf{Input:} }               & \multicolumn{4}{l}{\begin{tabular}[c]{@{}l@{}} $I$ Range feature vectors $\{\mathbf{f}_{\text{p}}^{i}\}_{i=1}^I$;\\ The position searching scope of the target;\\ The grid size $ \Delta M$;\\ The location of the $i$-th TBS $(x^i,y^i)$;\\ The location of PBS $(x^{\text{p}},y^{\text{p}})$.\end{tabular}} \\
\textbf{Output:}               & \multicolumn{4}{l}{The estimated location of target $(\hat{x}_{\text{tar}},\hat{y}_{\text{tar}})$.}\\ 
\multicolumn{5}{l}{\textbf{Stage 1: The aim is to obtain distance matching matrices.}}  \\
\multirow{-1}{*}{1:}        & \multicolumn{4}{l}{Obtain the searching vector in (\ref{eq34}) by} \\ & \multicolumn{4}{l}{ the position searching scope and $\Delta M$;}  \\
\multirow{-1}{*}{2:}        & \multicolumn{4}{l}{Obtain the distance compensation vector in (\ref{eq36}) by } \\
& \multicolumn{4}{l}{ the location of PBS $(x^{\text{p}},y^{\text{p}})$ and (\ref{eq34});}  \\
3:      & \multicolumn{4}{l}{$\textbf{for}$ $i$ in $I$ $\textbf{do}$}   \\
\multirow{-1}{*}{4:}        & \multicolumn{4}{l}{$\hspace{1em}$ Obtain the $i$-th distance vectors in (\ref{eq35}) by}\\
& \multicolumn{4}{l}{$\hspace{1em}$ the location of PBS $(x^i,y^i)$ and (\ref{eq34});}\\
\multirow{-1}{*}{5:}         & \multicolumn{4}{l}{$\hspace{1em}$ Obtain the $i$-th compensated distance vector in (\ref{eq37}) by}  \\
& \multicolumn{4}{l}{$\hspace{1em}$ (\ref{eq35}) and (\ref{eq36});}  \\
6:        & \multicolumn{4}{l}{$\hspace{1em}$ Obtain the $i$-th delay matching matrix in (\ref{eq38}) by (\ref{eq37});}  \\
7:       & \multicolumn{4}{l}{\textbf{end} \textbf{for}}   \\
\multicolumn{5}{l}{\textbf{Stage 2: The aim is to obtain the location of target.}} \\
8:                             & \multicolumn{4}{l}{Initialize a final position profile $\mathbf{p}_{\text{fn}}$; }  \\
9:                             & \multicolumn{4}{l}{\textbf{for} $i$ in $I$ \textbf{do}}  \\
10:                             & \multicolumn{4}{l}{$\hspace{1em}$ Obtain the $i$-th position profile in (\ref{eq39}) by}  \\ & \multicolumn{4}{l}{$\hspace{1em}$ the $i$-th range feature vector $\mathbf{f}_{\text{p}}^{i}$ and (\ref{eq38});}  \\
11:                             & \multicolumn{4}{l}{$\hspace{1em}$ Store the result in (\ref{eq39}) in $\mathbf{p}_{\text{fn}}$;}  \\
12:                            & \multicolumn{4}{l}{\textbf{end} \textbf{for}}   \\
13:                            & \multicolumn{4}{l}{Obtain the final position profile $\mathbf{p}_{\text{fn}}$=$\mathbf{p}_{\text{fn}}$/$I$;}   \\
14:                            & \multicolumn{4}{l}{Search the peak of $\mathbf{p}_{\text{fn}}$ to obtain peak index $\tilde{p_{\text{d}}}$;}   \\
15:                            & \multicolumn{4}{l}{Obtain the location of target $(\hat{x}_{\text{tar}},\hat{y}_{\text{tar}})$ by}  \\  & \multicolumn{4}{l}{the peak index $\tilde{p_{\text{d}}}$ and (\ref{eq34});}  \\
\hline
\end{tabular}}
\end{table}

\subsection{Velocity Estimation Method}
The velocity estimation method follows the same process as the localization method. The only difference is that no compensation is required because the derived Doppler frequency shift already includes the Doppler frequency shift from the TBS to target and from the target to PBS.

For the absolute velocity estimation of target, the scopes of the magnitude and angle of velocity are determined, denoted by $\left[V_\text{min}, V_\text{max}\right]$ and $\left[0, 2\pi\right]$, respectively. Then, the searching scope is gridded to obtain the velocity searching scope matrix, as shown in Fig.~\ref{fig3}.

\subsubsection{Obtain velocity searching scope matrix} \label{se4-B-1}
The searching scope of length $D \Delta \theta$ and width $S \Delta V$ is gridded with a step size of velocity's magnitude $\Delta V$ and a step size of velocity's direction $\Delta \theta$ to obtain a velocity searching scope matrix $\mathbf{V}_{\text{search}} \in \mathbb{C}^{S \times D}$.
\begin{equation} \label{eq41}
    {\footnotesize  \mathbf{V}_{\text{search}}=\left[
        \begin{array}{cccc}
         \left(v_{\text{ra},1},\theta_{\text{ra},1}\right)   & \left(v_{\text{ra},1},\theta_{\text{ra},2}\right) & \cdots &
         \left(v_{\text{ra},1},\theta_{\text{ra},D}\right)\\
          \left(v_{\text{ra},2},\theta_{\text{ra},1}\right)   & \left(v_{\text{ra},2},\theta_{\text{ra},2}\right) & \cdots &
         \left(v_{\text{ra},2},\theta_{\text{ra},D}\right)\\ 
         \vdots & \vdots & \ddots & \vdots  \\
         \left(v_{\text{ra},S},\theta_{\text{ra},1}\right)   & \left(v_{\text{ra},S},\theta_{\text{ra},2}\right) & \cdots &
         \left(v_{\text{ra},S},\theta_{\text{ra},D}\right)\\
        \end{array}\right],}
\end{equation}
where $\left(v_{\text{ra},s},\theta_{\text{ra},d}\right)$ represents the searching magnitude and angle of velocity with $s \in \{1,2,\cdots,S\}$ and $d \in \{1,2,\cdots,D\}$.

\subsubsection{Obtain velocity vectors and Doppler matrices} \label{se4-B-2}
$\mathbf{V}_{\text{search}}$ is transformed into a column vector to obtain a velocity searching vector $\mathbf{v}_{\text{search}} \in \mathbb{C}^{(S\times D) \times 1}$.
\begin{equation}  \label{eq42}
\mathbf{v}_{\text{search}}=\text{vec}
    \left(\mathbf{V}_{\text{search}}\right).
\end{equation}
Substituting the element from $\mathbf{v}_{\text{search}}$, AoA $\hat{\phi}_{\text{p},\text{ns}}^i$ and AoD $\hat{\theta}_{i,\text{ns}}$ obtained in Section \ref{se3-1},  into (\ref{eq9}) to obtain $I$ Doppler vectors, where the $i$-th Doppler vector is denoted by
\begin{equation} \label{eq42.2}
    \mathbf{s}_i=\left[f_{i,\text{p}}^{1}, f_{i,\text{p}}^{2}, \cdots, f_{i,\text{p}}^{S\times D}\right]^{\text{T}},
\end{equation}
where $f_{i,\text{p}}^{1}=\frac{-v_{\text{ra},1} f_\text{c}}c{\left[\cos\left(\theta_{\text{ra},1}-\hat{\phi}_{\text{p},\text{ns}}^i\right)+\cos\left(\theta_{\text{ra},1}-\hat{\theta}_{i,\text{ns}}\right)\right]}$.

According to the form of phase shifts generated by the Doppler frequency shift along the time axis in \cite{sturm2011waveform,wei2023carrier,liu2024carrier,ISAC_lht}, the $I$ Doppler vectors are transformed into $I$ Doppler matching matrices, where the $i$-th Doppler matching matrix $\mathbf{S}_{\text{mat}}^i \in \mathbb{C}^{(S\times D) \times ({M}_{\text{sym}}-1)}$ is expressed as
\begin{equation} \label{eq43}
{\fontsize{7}{8}\begin{aligned}
     \setlength{\arraycolsep}{1pt} 
    &\mathbf{S}_{\text{mat}}^i= \\& \left[
     \begin{array}{ccccc}
      e^{j2\pi T f_{i,\text{p}}^{1}}  & \cdots & e^{j2\pi m' T f_{i,\text{p}}^{1}} & \cdots & e^{j2\pi (M_{\text{sym}}-1) T f_{i,\text{p}}^{1}}\\    
      \vdots & \vdots & \ddots & \vdots \\
      e^{j2\pi T f_{i,\text{p}}^{z}}  & \cdots & e^{j2\pi m' T f_{i,\text{p}}^{z}} & \cdots & e^{j2\pi (M_{\text{sym}}-1) T f_{i,\text{p}}^{z}}\\ 
      \vdots & \vdots & \ddots & \vdots \\
      e^{j2\pi T f_{i,\text{p}}^{(S\times D)}}  & \cdots & e^{j2\pi m' T f_{i,\text{p}}^{(S\times D)}} & \cdots & e^{j2\pi (M_{\text{sym}}-1) T f_{i,\text{p}}^{(S\times D)}}
     \end{array}
     \right] ,
\end{aligned}}
\end{equation}
where $z \in \{1,2,\cdots,(S\times D)\}$.

\subsubsection{Conjugate inner product}
We obtain $I$ Doppler matching matrices in (\ref{eq43}) and $I$ velocity feature vectors in (\ref{eq32}). Then, they are utilized to estimate the absolute velocity of target.

For the $i$-th $\mathbf{S}_{\text{mat}}^i$ and the corresponding $\mathbf{e}_{\text{p}}^i$, the conjugate inner product operation is replaced by matrix multiplication. Then, the $i$-th velocity profile $\mathbf{v}_i$ after matrix multiplication is expressed in (\ref{eq45}).
\begin{figure*}
\begin{equation}\label{eq45}
    \begin{aligned}
&\mathbf{v}_i=\mathbf{S}_i^\text{mat}\left(\mathbf{e}_\text{p}^i\right)^\text{H} \\
&=\left[\sum_{m'=1}^{M_{\text{sym}}-1}e^{j2\pi m'T\left(f_{i,\text{p}}^{1}-f_{i,\text{p}}\right)},\sum_{m'=1}^{M_{\text{sym}}-1}e^{j2\pi m'T\left(f_{i,\text{p}}^{2}-f_{i,\text{p}}\right)},\cdots,\sum_{m'=1}^{M_{\text{sym}}-1}e^{j2\pi m'T\left(f_{i,\text{p}}^{(S\times D)}-f_{i,\text{p}}\right)}\right]^{\text{T}}\in\mathbb{C}^{(S\times D)\times1}.
\end{aligned}
\end{equation}
\end{figure*}

Then, the $I$ velocity profiles are coherently accumulated to obtain the final velocity profile $\mathbf{v}_{\text{fn}}$.
\begin{equation}\label{eq46}
    \mathbf{v}_\text{fn}=\frac1I\sum_{i=1}^I\mathbf{v}_i,
\end{equation}
the peak value of which is searched to obtain the peak index value $\tilde{z}$. The value of $\mathbf{v}_{\text{search}}(\tilde{z})$ is the final estimated absolute velocity of target, denoted by $(\hat{v},\hat{\theta})$.

\hyperref[tab4]{\textbf{Algorithm 3}} demonstrates the velocity estimation method in the symbol-level fusion method of multi-BS sensing information.
\begin{table}[ht]
\centering
\label{tab4}
\resizebox{0.5\textwidth}{!}{
\setlength{\arrayrulewidth}{1.5pt}
\begin{tabular}{rllll}
\hline
\multicolumn{5}{l}{\textbf{Algorithm 3:} Velocity estimation method}   \\ \hline
\multirow{-4}{*}{\textbf{Input:} }               & \multicolumn{4}{l}{\begin{tabular}[c]{@{}l@{}} $I$ Velocity feature vectors $\{\mathbf{e}_{\text{p}}^{i}\}_{i=1}^I$;\\ The velocity searching scope of the target;\\ The grid size $ \Delta \theta$ and $\Delta V$;\\ The estimated AoA $\hat{\phi}_{\text{p},\text{ns}}^i$ and AoD $\hat{\theta}_{i,\text{ns}}$. \end{tabular}} \\
\textbf{Output:}               & \multicolumn{4}{l}{The estimated absolute velocity of target $(\hat{v},\hat{\theta})$.}\\ 
\multicolumn{5}{l}{\textbf{Stage 1: The aim is to obtain velocity matching matrices.}}  \\
\multirow{-1}{*}{1:}        & \multicolumn{4}{l}{Obtain the searching vector in (\ref{eq42}) by}\\ & \multicolumn{4}{l}{the velocity searching scope, the grid size $ \Delta \theta$ and $\Delta V$;}  \\
2:      & \multicolumn{4}{l}{$\textbf{for}$ $i$ in $I$ $\textbf{do}$}   \\
\multirow{-1}{*}{3:}         & \multicolumn{4}{l}{$\hspace{1em}$ Obtain the $i$-th Doppler vector in (\ref{eq42.2}) by}\\
& \multicolumn{4}{l}{$\hspace{1em}$ the estimated AoA $\hat{\phi}_{\text{p},\text{ns}}^i$, AoD $\hat{\theta}_{i,\text{ns}}$, and (\ref{eq9});}\\
4:        & \multicolumn{4}{l}{$\hspace{1em}$ Obtain the $i$-th Doppler matching matrix in (\ref{eq43}) by}  \\ & \multicolumn{4}{l}{$\hspace{1em}$ (\ref{eq42.2});}  \\
5:       & \multicolumn{4}{l}{\textbf{end} \textbf{for}}   \\
\multicolumn{5}{l}{\textbf{Stage 2: The aim is to obtain the absolute velocity of target.}} \\
6:                             & \multicolumn{4}{l}{Initialize a final velocity profile $\mathbf{v}_{\text{fn}}$; }  \\
7:                             & \multicolumn{4}{l}{\textbf{for} $i$ in $I$ \textbf{do}}  \\
8:                             & \multicolumn{4}{l}{$\hspace{1em}$ Obtain the $i$-th velocity profile in (\ref{eq45}) by}  \\  & \multicolumn{4}{l}{$\hspace{1em}$ the $i$-th velocity feature vector $\mathbf{e}_{\text{p}}^{i}$ and (\ref{eq43});}  \\
9:                             & \multicolumn{4}{l}{$\hspace{1em}$ Store the result in (\ref{eq45}) in $\mathbf{v}_{\text{fn}}$;}  \\
10:                            & \multicolumn{4}{l}{\textbf{end} \textbf{for}}   \\
11:                            & \multicolumn{4}{l}{Obtain the final velocity profile $\mathbf{v}_{\text{fn}}$=$\mathbf{v}_{\text{fn}}$/$I$;}   \\
12:                            & \multicolumn{4}{l}{Search the peak of $\mathbf{v}_{\text{fn}}$ to obtain peak index $\tilde{z}$;}   \\
13:                            & \multicolumn{4}{l}{Obtain the absolute velocity of target $(\hat{v},\hat{\theta})$ by}  \\ & \multicolumn{4}{l}{ the peak index $\tilde{z}$ and (\ref{eq42}) ;}  \\
\hline
\end{tabular}}
\end{table}

\section{Performance Analysis of Multi-BS Cooperative Passive Sensing} \label{se5}
In this section, the computational complexity of the joint AoA and AoD estimation method and the traditional 2D-MUSIC method, as well as the SNR in the symbol-level fusion method of multi-BS sensing information, are analyzed.

\subsection{Complexity of Joint AoA and AoD Estimation Method}
\subsubsection{Joint AoA and AoD estimation method}
The computational complexity of 2D-MUSIC method for MIMO radar is derived in~\cite{zhang2010direction} as follows.
\begin{equation} \label{eq47}
    \mathcal{O}\{LM^2N^2+M^3N^3+\Theta^2[MN(MN-K)+MN-K]\},
\end{equation}
where $K$ and $L$ denote the number of targets and snapshots, respectively; $M$ and $N$ denote the number of transmitting and receiving antennas, respectively; $\Theta$ denotes the total time spent in searching a dimension, where ``dimension'' refers to the searching scope for an unknown parameter. Inspired by (\ref{eq47}), the computational complexity of the proposed method is derived as follows.

In the rough estimation stage, 2D-FFT is applied to estimate the AoA and AoD. The computational complexity of 2D-FFT is 
\begin{equation} \label{eq48}
    \mathcal{O}\left\{N_{\text{Rx}}^{\text{p}}N_{\text{Tx}}^{i}\left[\left(N_{\text{Rx}}^{\text{p}}\right)^2+\left(N_{\text{Tx}}^{i}\right)^2\right]\right\}.
\end{equation}
In the fine estimation stage, the searching scope is narrowed down as in (\ref{eq18}). Then, the narrowed searching scope is gridded and 
the searching values are sequentially substituted into (\ref{eq22}) to obtain the estimation of AoA and AoD. The grid sizes of AoA and AoD are expressed as follows. 
\begin{equation} 
\begin{aligned} \label{eq51}
&\phi_{\text{step}}=\frac{\arcsin\left(\frac{\lambda\left(\tilde{\mu}_\phi+1\right)}{d_\text{r}N_\text{Rx}^\text{p}}\right)-\arcsin\left(\frac{\lambda\left(\tilde{\mu}_\phi-1\right)}{d_\text{r}N_\text{Rx}^\text{p}}\right)}{{\epsilon_{\phi}}}, \\
&\theta_{\text{step}}=\frac{\arcsin\left(\frac{\lambda\left(\tilde{\mu}_{\theta}+1\right)}{d_{\text{r}}N_{\text{Tx}}^{i}}\right)-\arcsin\left(\frac{\lambda\left(\tilde{\mu}_{\theta}-1\right)}{d_{\text{r}}N_{\text{Tx}}^{i}}\right)}{\epsilon_{\theta}},
\end{aligned}
\end{equation}
where $\epsilon_{\phi}$ and $\epsilon_{\theta}$ denote the number of grids for the estimation of AoA and AoD, respectively.
Meanwhile, the searching complexity of 2D-MUSIC is related to the number of grids in searching scope. Therefore, the searching complexity of 2D-MUSIC is $\mathcal{O}\{\epsilon_{\phi}\epsilon_{\theta}\}$ and the total computational complexity of the fine estimation stage is
\begin{equation} \label{eq49}
    \begin{aligned}
     \mathcal{O}&\left\{M_\text{sym}N_\text{c}\left(N_\text{Rx}^\text{p}N_\text{Tx}^{i}\right)^2+\left(N_\text{Rx}^\text{p}N_\text{Tx}^{i}\right)^3 \right. \\  & \left. +\epsilon_{\phi}\epsilon_{\theta}\left[\left(N_\text{Rx}^\text{p}N_\text{Tx}^{i}+1\right)\left(N_\text{Rx}^\text{p}N_\text{Tx}^{i}-1\right)\right]\right\}.
    \end{aligned} 
\end{equation}

Thus, the total computational complexity of the low-complexity joint AoA and AoD estimation method is shown in (\ref{eq50}).
\begin{figure*}[htbp] 
    \begin{equation} \label{eq50}
        \mathcal{O}\left\{\underbrace{N_{\text{Rx}}^{\text{p}}N_{\text{Tx}}^{i}\left[\left(N_{\text{Rx}}^{\text{p}}\right)^2+\left(N_{\text{Tx}}^{i}\right)^2\right]}_{\text{Rough stage}}+ \underbrace{M_\text{sym}N_\text{c}\left(N_\text{Rx}^\text{p}N_\text{Tx}^{i}\right)^2+\left(N_\text{Rx}^\text{p}N_\text{Tx}^{i}\right)^3 +\epsilon_{\phi}\epsilon_{\theta}\left[\left(N_\text{Rx}^\text{p}N_\text{Tx}^{i}+1\right)\left(N_\text{Rx}^\text{p}N_\text{Tx}^{i}-1\right)\right]}_{\text{Fine stage}}  \right\}.
    \end{equation}  
    {\noindent} \rule[-10pt]{18cm}{0.1em}
\end{figure*}

\subsubsection{Traditional 2D-MUSIC method}
In deriving the computational complexity of the traditional 2D-MUSIC method, the grid sizes of the searching scope are set to $\phi_{\text{step}}$ and $\theta_{\text{step}}$ to ensure consistent estimation accuracy between traditional 2D-MUSIC method and proposed method.
Meanwhile, the searching scope of both AoA and AoD are $\left(0,\ \pi\right]$, and the searching complexity of the traditional 2D-MUSIC is
$\mathcal{O}\left\{\gamma_\phi\gamma_\theta\right\}$, where $\gamma_\phi=\pi/\phi_\text{step}$ and $\gamma_\theta=\pi/\theta_\text{step}$.

Therefore, the computational complexity of traditional 2D-MUSIC method is derived based on (\ref{eq47}) as
\begin{equation} \label{eq52}
\begin{aligned}
     \mathcal{O}&\left\{\underbrace{M_{\text{sym}}N_{\text{c}}\left(N_{\text{Rx}}^{\text{p}}N_{\text{Tx}}^{{i}}\right)^2}_{(\text{a})}+\underbrace{\left(N_{\text{Rx}}^{\text{p}}N_{\text{Tx}}^{{i}}\right)^3}_{(\text{b})} \right. \\ & \left. +\underbrace{\gamma_{\phi}\gamma_{\theta}\left[\left(N_{\text{Rx}}^{\text{p}}N_{\text{Tx}}^{{i}}+1\right)(N_{\text{Rx}}^{\text{p}}N_{\text{Tx}}^{{i}}-1)\right]}_{(\text{c})}\right\},
\end{aligned}      
\end{equation}
where (a) denotes the computational complexity of constructing covariance matrix, (b) denotes the computational complexity of eigenvalue decomposition \cite{trefethen2022numerical}, and (c) denotes the computational complexity of calculating and searching spatial spectral function. 

\subsection{SNR gain of symbol-level fusion method of multi-BS sensing information}\label{se5-B} 
SNR is an important performance metric for sensing. In this section, the SNR gain of the proposed symbol-level fusion method of multi-BS sensing information is derived in \hyperref[the_3]{\textbf{Theorem 3}}.
\begin{theorem}\label{the_3}
    The SNR gain of the proposed symbol-level fusion method of multi-BS sensing information is
    \begin{equation}      \begin{cases}G_\mathrm{P}=\left(N_\text{c}-1\right)I\quad\left(\text{in location estimation}\right),\\G_\text{V}=\left(M_\text{sym}-1\right)I\quad\left(\text{in velocity estimation}\right).\end{cases}
    \end{equation}
\end{theorem}
\begin{proof}
    Assuming that the modulus of the $n'$-th element of the $i$-th range feature vector $\mathbf{f}_{\text{p}}^{i}$ obtained in Section \ref{se3} is $h_i$ and the variance of the noise is $snr_i=h_i^2\sigma^2$, the SNR of $\mathbf{f}_{\text{p}}^{i}$ is $1/\sigma^2$. In the localization method, the operation in Sections \ref{se4-1-1}, \ref{se4-1-2} and \ref{se4-1-3} will not affect the SNR of $\mathbf{f}_{\text{p}}^{i}$. Therefore, the SNR gain originates solely from the operation of conjugate inner product.

    Assume that the module of the $i$-th delay matching matrix $\mathbf{G}_{\text{mat}}^i$ is 1. 
    For (\ref{eq39}), the elements of the $i$-th position profile $\mathbf{p}_i$ are obtained by multiplying and summing each column of $\mathbf{G}_{\text{mat}}^i$ and $\mathbf{f}_{\text{p}}^{i}$ element by element. After the two elements are multiplied, as the module of $\mathbf{G}_{\text{mat}}^i$ is 1, the obtained result is the SNR of $\mathbf{f}_{\text{p}}^{i}$, which is $1/\sigma^2$. After the multiplication results are summed, the modules of signal and noise are $\left(N_{\text{c}}-1\right)h_i$ and $\left(N_{\text{c}}-1\right)h_i^2\sigma^2$, respectively. Therefore, the SNR of $\mathbf{p}_i$ is
    \begin{equation}
        snr_\text{c}^i=\frac{\left(N_\text{c}-1\right)^2h_i^2}{\left(N_\text{c}-1\right)h_i^2\sigma^2}=\frac{\left(N_\text{c}-1\right)}{\sigma^2}.
    \end{equation}

    (\ref{eq40}) is a coherent accumulation operation. According to \cite{richards2010noncoherent}, the SNR gain of coherent accumulation is related to the number of summation terms. Therefore, the SNR of the final position profile $\mathbf{p}_{\text{fn}}$ is
    \begin{equation}
        snr_{\mathrm{f}}=\frac{\left(N_{\text{c}}-1\right)I}{\sigma^2},
    \end{equation}
    and the SNR gain of the localization method is
    \begin{equation} \label{eq55}
        G_{\text{P}}=\frac{snr_{\text{f}}}{snr_i}=\left(N_{\text{c}}-1\right)I.
    \end{equation}
    Similarly, the SNR gain of the velocity estimation method is $G_\text{V}=\left(M_\text{sym}-1\right)I$.
\end{proof}

In summary, the proposed symbol-level fusion method of multi-BS sensing information obtains the SNR gain compared with  single-BS sensing, where the SNR gain is increasing with the increase of the number of TBSs.

\begin{table*}[!htbp]
	\caption{ Simulation parameters \cite{sturm2009ofdm,sturm2011waveform,wei2023carrier,wei2023symbol}}
	\label{tab_simulation}
	\renewcommand{\arraystretch}{1.3} 
	\begin{center}\resizebox{0.9\linewidth}{!}{
		\begin{tabular}{|m{0.07\textwidth}<{\centering}| m{0.32\textwidth}<{\centering}| m{0.13\textwidth}<{\centering}| m{0.07\textwidth}<{\centering}|m{0.32\textwidth}<{\centering}| m{0.13\textwidth}<{\centering}|}
			\hline
			\textbf{Symbol} & \textbf{Parameter} & \textbf{Value} & \textbf{Symbol} & \textbf{Parameter} & \textbf{Value} \\
			\hline
			$N_\text{c}$	& Number of subcarriers & $512$ & $M_{\text{sym}}$	& Number of symbols  & $256$  \\
			\hline
			$f_\text{c}$	& Carrier frequency & 24 GHz & $\Delta f$	& Subcarrier spacing  & 120 kHz \\
			\hline
			$T_{\text{ofdm}}$	& Elementary NC-OFDM symbol duration  & 8.33 $\rm {\mu s}$ & $T_{\text{p}}$	& Cyclic prefix length & 1.33 $\rm {\mu s}$ \\
			\hline
                $T$	& Entire NC-OFDM symbol duration  & 9.66 $\rm {\mu s}$  & $I$ & Number of TBS & $1\sim 4$  \\
               \hline
			$N_{\text{Tx}}^i $	& Number of transmitting antenna in TBS & $64$  & $N_{\text{Rx}}^\text{p} $	& Number of receiving antenna in TBS & $64$ \\
			\hline
			$(x^1,y^1)$	& The coordinate of the $1$-th TBS & (40 m, 0 m) &  $(x^2,y^2)$ & The coordinate of the $2$-th TBS & (0 m, 40 m) \\			
			\hline
		   $(x^3,y^3)$	& The coordinate of the $3$-th TBS & (0 m, 80 m)  &   $(x^4,y^4)$	& The coordinate of the $4$-th TBS & (80 m, 0 m)  \\
                \hline
              $(x^{\text{p}},y^{\text{p}})$	& The coordinate of PBS & (80 m, 80 m)  &  $(x_{\text{tar}},y_{\text{tar}})$	& The coordinate of target & (40 m, 40 m) \\
			\hline
             $v$	& Magnitude of velocity of the target & 27 m/s & $\theta$& Angle of velocity of the target & 0.785 radians  \\ \hline 
             $\xi_{ \tau,i}$ & Value of TOs & $\left[30,60\right]$ ns & $\xi_{ f,i}$ & Value of CFOs & $\left[0.03,0.06\right]\ \Delta f$ \\ 
			\hline
		\end{tabular}}
	\end{center}
\end{table*}

\section{Simulation} \label{se6}

In this section, the proposed multi-BS cooperative passive sensing method is
evaluated. Additionally, we compare the computational complexity of the proposed 
joint AoA and AoD estimation method with the traditional 2D-MUSIC method.

\subsection{Multi-BS Cooperative Passive Sensing}

In this section, the position and velocity profiles with symbol-level fusion method of sensing information using three TBSs are simulated.
Then, the root mean square errors (RMSEs) of location and velocity estimation are simulated. 
Table~\ref{tab_simulation} shows the main simulation parameters, 
which satisfies the requirement of IoV \cite{sturm2009ofdm,sturm2011waveform,wei2023carrier,wei2023symbol}. The simulation results are obtained with 10000 times Monte Carlo simulations.

\subsubsection{Position profiles}

To evaluate the feasibility of the proposed location estimation method, 
the simulation results are shown in Fig.~\ref{fig4} 
and the white lines represent contour lines.

Observing (a), (b), and (c) of Fig. \ref{fig4}, 
the trajectories of peak index values form a curve, 
which does not allow for accurate location estimation of target. 
Therefore, we can coherently accumulate multiple position profiles 
in (\ref{eq40}) to obtain location estimation of target, 
resembling finding intersection point of multiple curves.

Fig.~\ref{fig4(d)} shows the the final position profile 
obtained by performing coherent accumulation operation to the position profiles of three TBSs. 
According to (d) of Fig.~\ref{fig4}, the location estimation of the target is (39.99 m, 39.99 m). 
The error of target location estimation in x-axis direction is 0.01 m, which reveals the feasibility of location estimation of the proposed 
multi-BS cooperative passive sensing method. 

\begin{figure}[htbp]
\centering
\subfigure[Position profile of the 1-th TBS]{\includegraphics[width=0.48\linewidth]{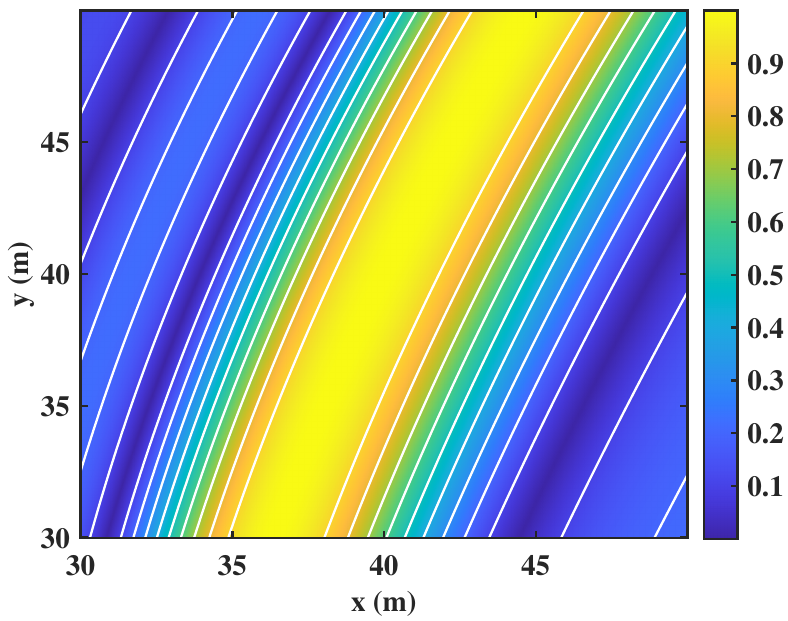}
\label{fig4(a)}}
\subfigure[Position profile of the 2-th TBS]{\includegraphics[width=0.48\linewidth]{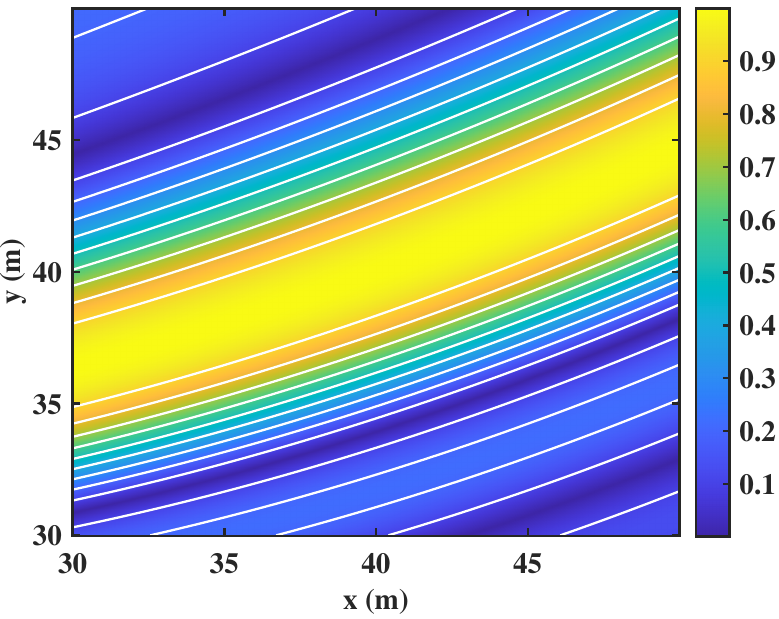}
\label{fig4(b)}}
\vspace*{0.5cm}
\subfigure[Position profile of the 3-th TBS]{\includegraphics[width=0.48\linewidth]{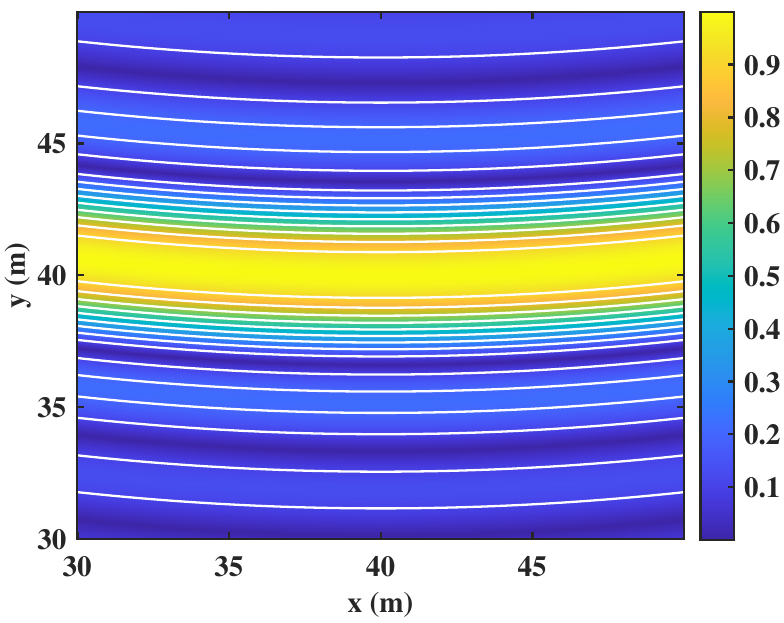}
\label{fig4(c)}}
\subfigure[The final position profile in (\ref{eq40})]{\includegraphics[width=0.48\linewidth]{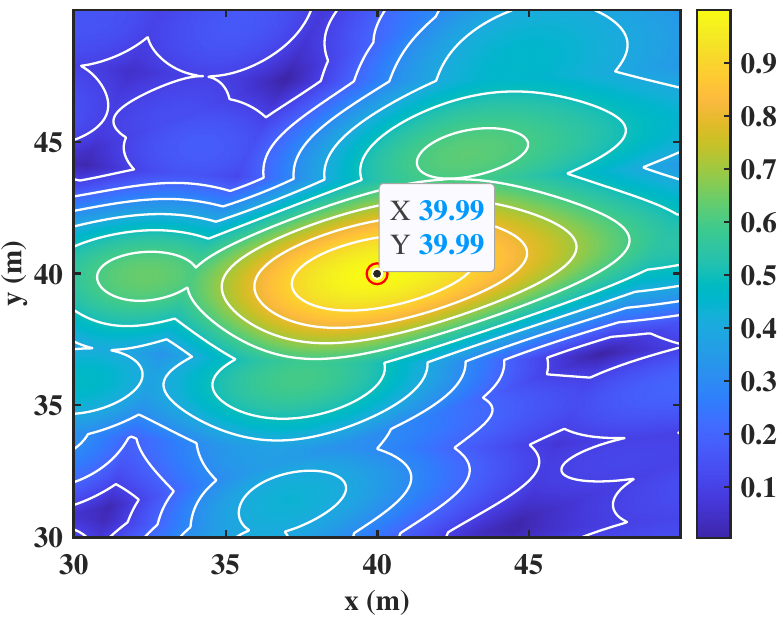}
\label{fig4(d)}}
\caption{Position profiles of location estimation with SNR = -5dB}
\label{fig4}
\end{figure}

\subsubsection{Velocity profiles}

To evaluate the feasibility of the proposed velocity estimation method, the simulation results are shown in Fig.~\ref{fig5}.

The meaning of the peak value in Fig. \ref{fig5} is that the horizontal axis value corresponding to the peak value is the estimated angle of target, and the vertical axis value is the estimated magnitude of target. However, the trajectories of index values corresponding to the peak values in (a), (b), and (c) of Fig. \ref{fig5} form a curve, which does not allow for accurate absolute velocity estimation of target. 
Therefore, we can coherently accumulate multiple velocity profiles in (\ref{eq46}) to obtain the absolute velocity estimation of target, resembling finding intersection point of multiple curves.

As illustrated in (d) of Fig.~\ref{fig5}, the peak corresponds to a velocity magnitude estimation of 27.3 m/s and an angle estimation of 0.754 radians. The errors in magnitude and angle estimation are 0.3 m/s and 0.03 radians, respectively, verifying the feasibility of absolute velocity estimation of the proposed multi-BS cooperative passive sensing method.

\begin{figure}[htbp]
\centering
\subfigure[Velocity profile of the 1-th TBS]{\includegraphics[width=0.48\linewidth]{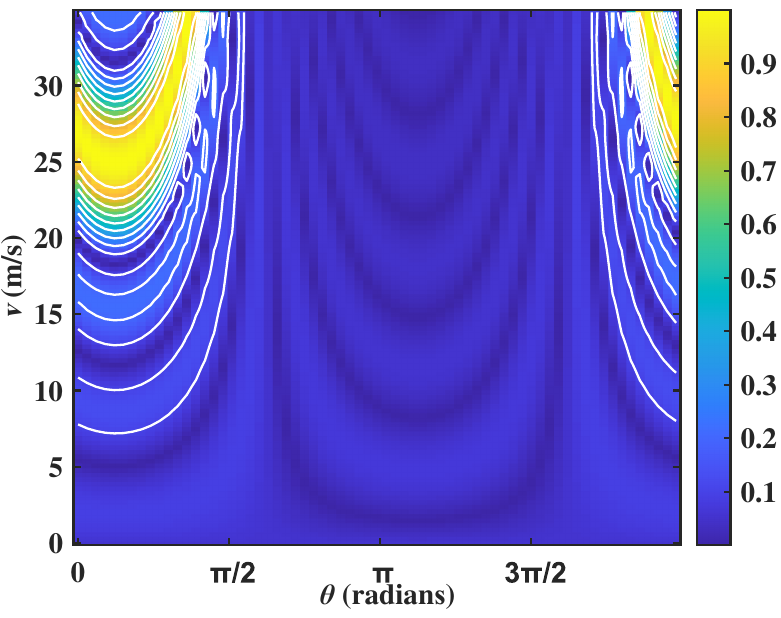}
\label{fig5(a)}}
\subfigure[Velocity of the 2-th TBS]{\includegraphics[width=0.48\linewidth]{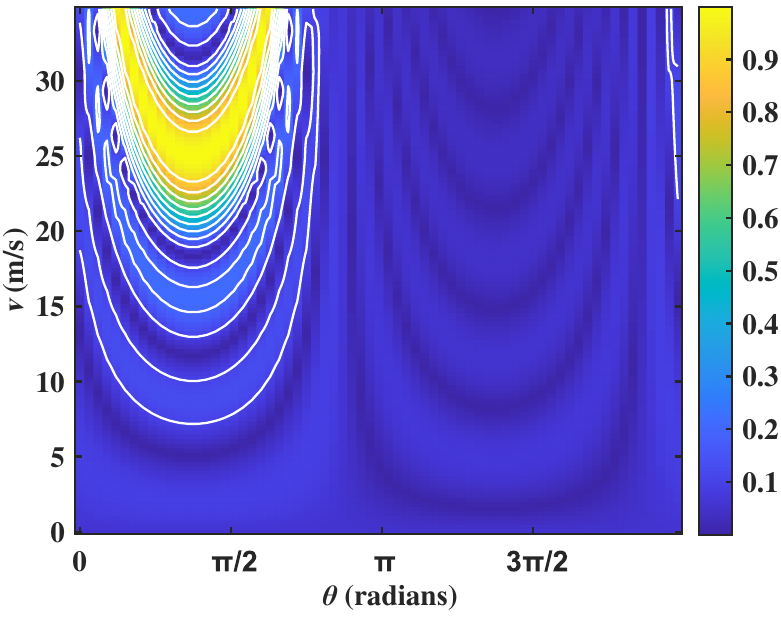}
\label{fig5(b)}}
\vspace*{0.5cm}
\subfigure[Velocity of the 3-th TBS]{\includegraphics[width=0.48\linewidth]{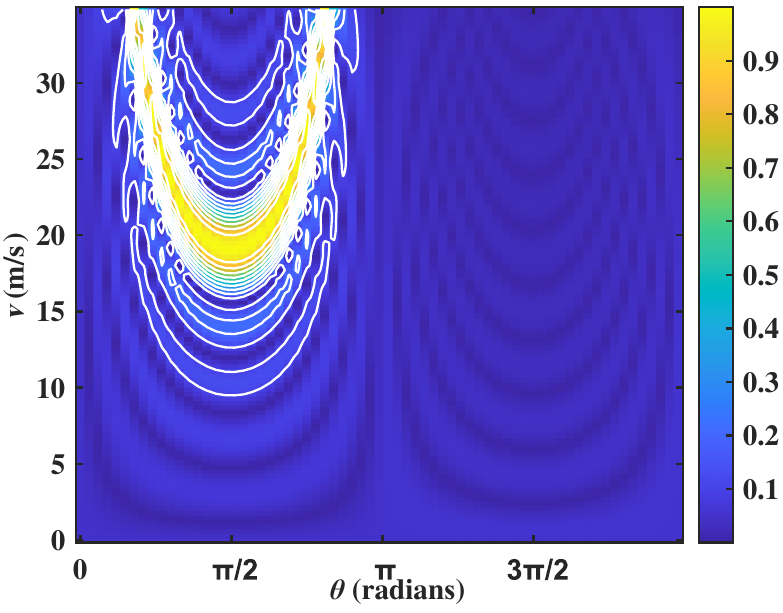}
\label{fig5(c)}}
\subfigure[The final velocity profile in (\ref{eq46})]{\includegraphics[width=0.48\linewidth]{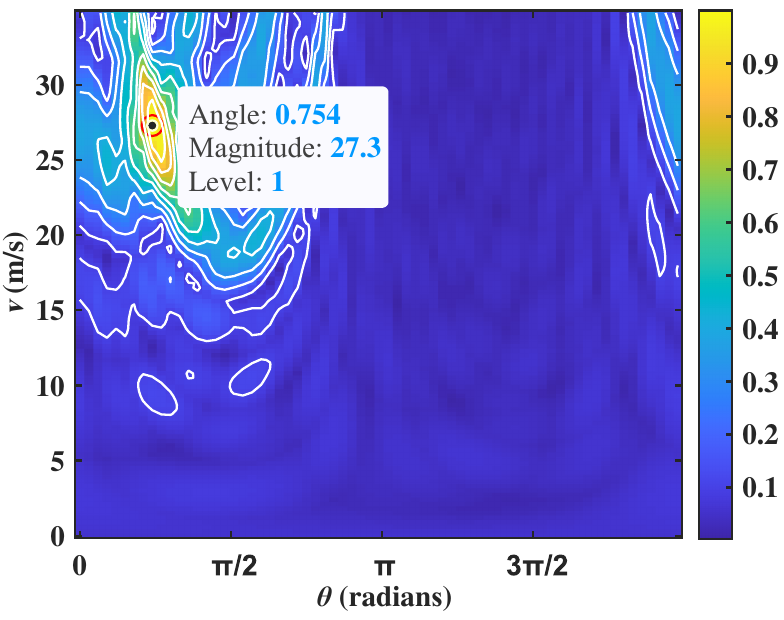}
\label{fig5(d)}}
\caption{Velocity profiles of velocity estimation with SNR = -5dB}
\label{fig5}
\end{figure}

\subsubsection{RMSE of location estimation}

The RMSE of location estimation with different numbers of TBSs 
are simulated as shown in Fig.~\ref{fig6}. For a single TBS passive sensing, we use the traditional AoA localization method~\cite{liu2024target}.
Observing the red, blue, and green lines in Fig.~\ref{fig6}, 
when the SNR is higher than -10 dB, 
the accuracy of location estimation of target reaches the centimeter level, revealing that the proposed multi-BS cooperative passive sensing scheme 
has high-accuracy location estimation of target.
As the number of TBSs increases, 
the lines gradually decline, 
verifying the superiority of the proposed multi-BS cooperative passive sensing scheme compared with single-TBS passive sensing.

The performance of the proposed NLCC method is further simulated under different values of TOs, as shown in Fig.~\ref{fig7}. 
Since the CFO has limited impact on the location estimation of target, 
a constant value is assigned when analyzing the performance of location estimation~\cite{jiang2023cooperation}. 
The solid and dashed lines refer to the cooperative passive sensing with and without NLCC, respectively. 
Cooperative passive sensing without NLCC refers to 
performing coherent compression directly using the delay-Doppler 
NLoS information matrices.
The dotted line refers to the cooperative passive sensing 
without synchronization error.

Firstly, as the value of TO increases, 
the dashed lines rise, indicating that the impact of TO on the performance of location estimation is related to the value of TO. 
Secondly, the solid lines are lower and almost overlapped  
compared to the dashed lines, 
indicating that the proposed NLCC method can mitigate the degradation 
of location estimation performance caused by TO. 
Finally, the solid lines closely approximate the dotted line in high SNR region, 
indicating that the performance of NLCC method 
is improving with the increase of SNR.

\begin{figure}
    \centering
    \includegraphics[width=0.45\textwidth]{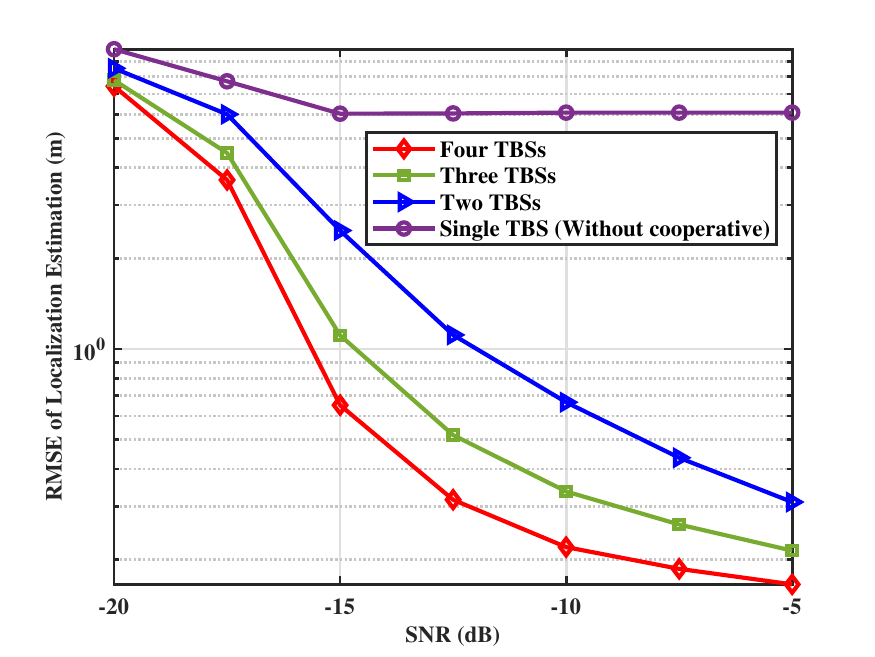}
    \caption{RMSE of location estimation for multi-BS cooperative passive sensing with TO = 30 ns and CFO = 0.03$\Delta f$}
    \label{fig6}
\end{figure}
\begin{figure}
    \centering
    \includegraphics[width=0.45\textwidth]{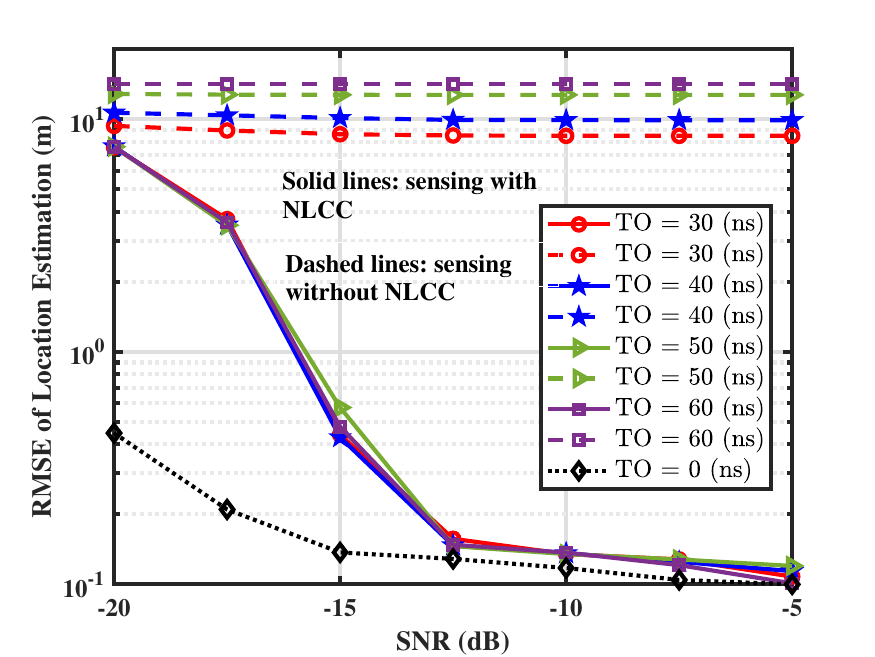}
    \caption{RMSE of location estimation with different TOs and CFO = 0.03$\Delta f$}
    \label{fig7}
\end{figure}

\subsubsection{RMSE of velocity estimation} 
The average RMSE of the magnitude and angle of velocity serves as the RMSE of velocity estimation in simulation.
Absolute velocity estimation of target remains infeasible with single TBS passive sensing and is excluded from the simulation.
Similarly, the RMSE of cooperative passive sensing with and without NLCC are compared.
As shown in Figs.~\ref{fig8} and~\ref{fig9}, the conclusions are similar to location estimation.

\begin{figure}
    \centering
    \includegraphics[width=0.45\textwidth]{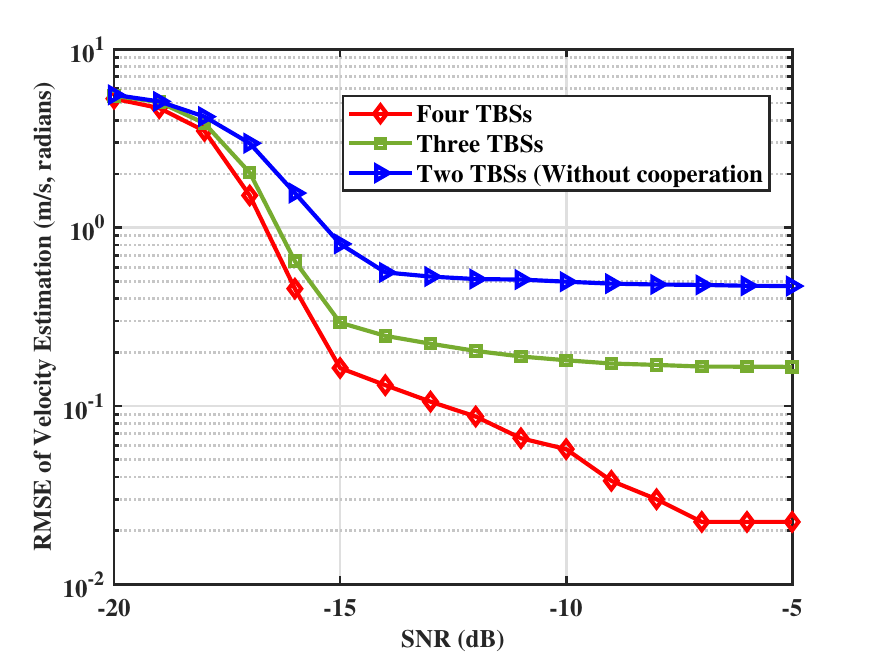}
    \caption{RMSE of velocity estimation for multi-BS cooperative passive sensing with TO = 30 ns and CFO = 0.03$\Delta f$}
    \label{fig8}
\end{figure}
\begin{figure}
    \centering
    \includegraphics[width=0.45\textwidth]{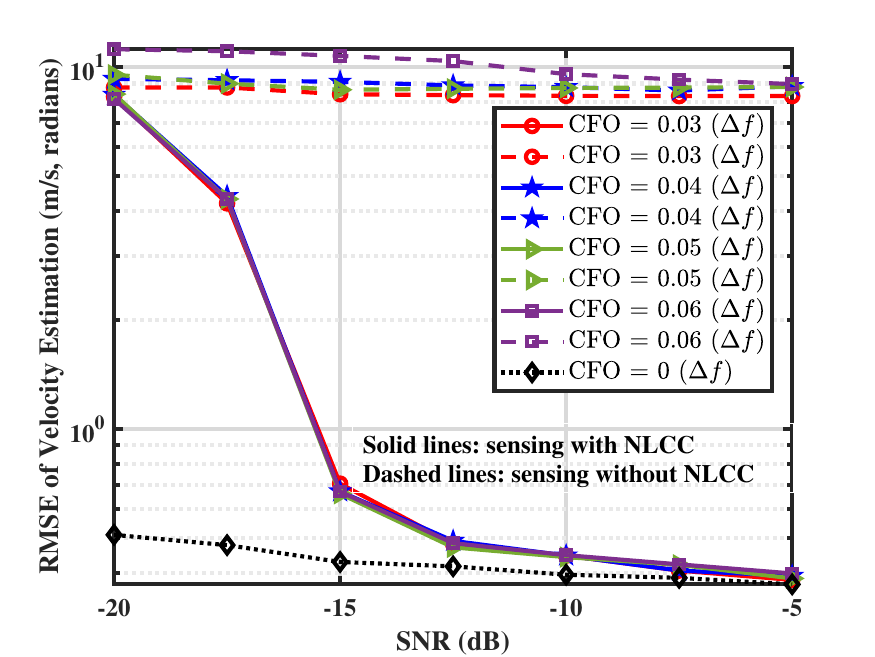}
    \caption{RMSE of velocity estimation with different CFOs and TO = 30 ns}
    \label{fig9}
\end{figure}

\subsection{Joint AoA and AoD Estimation}\label{se6-A}
In this section, the computational complexity of the proposed joint AoA and AoD estimation method and traditional 2D-MUSIC method are simulated while ensuring the same estimation accuracy between the two methods.
The estimation accuracy depends solely on the grid size of the fine estimation stage, thereby employing the same grid size $\phi_{\text{step}}$ = $\theta_{\text{step}}$ = 0.01 radians for both methods ensures the same estimation accuracy.
According to (\ref{eq50}) and (\ref{eq51}), the analytical simulation is revealed in Fig.~\ref{fig10}, with the number of received antennas as the independent variable.

\begin{figure}
    \centering
    \includegraphics[width=0.45\textwidth]{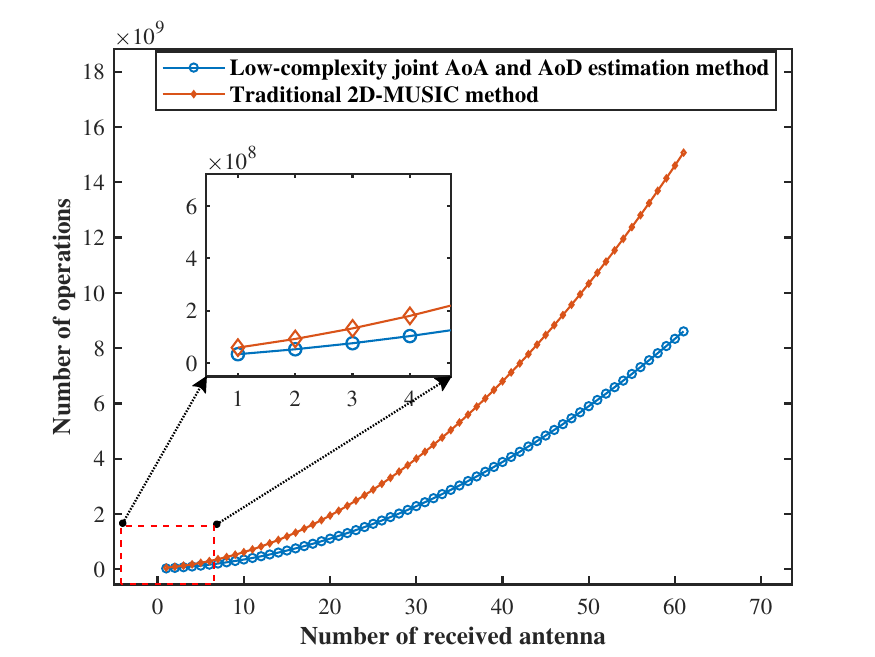}
    \caption{Computational complexity of the two methods}
    \label{fig10}
\end{figure}

As depicted in Fig. \ref{fig10}, the computational complexity of both methods increases rapidly with the increase of the number of received antennas. With the same number of received antennas, the computational complexity of the proposed method is lower than the traditional 2D-MUSIC method.
Furthermore, the disparity in computational complexity between the two methods increases as the number of received antennas increases.
Besides, we operate the two methods on the simulation equipment where the CPU parameter is 13-th Gen Inter (R) Core (TM) i9-13900KH, and the running time is illustrated in Table \ref{tab_runing}.

\begin{table}[ht]
\caption{The running time of two methods with $N_\text{Rx}^\text{p}$ = 64}
\label{tab_runing}
\renewcommand{\arraystretch}{1.5} 
\begin{center}
\resizebox{.48\textwidth}{!}{
\begin{tabular}{|ccc|}
\hline
\multicolumn{3}{|c|}{Computational complexity}       \\ \hline
\multicolumn{1}{|c|}{Method}        & \multicolumn{1}{|c|}{Our proposed method}  & 2D-MUSIC~\cite{pasya2014joint} \\ \hline
\multicolumn{1}{|c|}{Run Time (s)} & \multicolumn{1}{c|}{$0.33619$}   & \multicolumn{1}{c|}{$3.6458$}   \\ \hline
\end{tabular}}
\end{center}
\end{table}

\section{Conclusion} \label{se7}
In this paper, a multi-BS cooperative passive sensing scheme is proposed to obtain high-accuracy sensing performance. To overcome the synchronization error in passive sensing, the NLCC method is proposed to mitigate the CFO and TO by correlating the sensing information in NLoS path and LoS path. Then, a symbol-level fusion method of multi-BS sensing information is proposed to achieve high-accuracy location and absolute velocity estimation of target. Moreover, we propose a low-complexity joint AoA and AoD estimation method with rough and fine estimation stages to achieve high-accuracy AoA and AoD estimation with low complexity. The simulation results demonstrate the superiority of the proposed multi-BS cooperative passive sensing scheme, as well as the effectiveness of the proposed NLCC method.
The proposed fusion method does not yet utilize multi-BS sensing information. Future work will explore deep fusion and AI fusion in multi-BS cooperative passive sensing.

\newpage
\bibliographystyle{IEEEtran}
\bibliography{reference}

\begin{thebibliography}{10}
\providecommand{\url}[1]{#1}
\csname url@samestyle\endcsname
\providecommand{\newblock}{\relax}
\providecommand{\bibinfo}[2]{#2}
\providecommand{\BIBentrySTDinterwordspacing}{\spaceskip=0pt\relax}
\providecommand{\BIBentryALTinterwordstretchfactor}{4}
\providecommand{\BIBentryALTinterwordspacing}{\spaceskip=\fontdimen2\font plus
\BIBentryALTinterwordstretchfactor\fontdimen3\font minus \fontdimen4\font\relax}
\providecommand{\BIBforeignlanguage}[2]{{%
\expandafter\ifx\csname l@#1\endcsname\relax
\typeout{** WARNING: IEEEtran.bst: No hyphenation pattern has been}%
\typeout{** loaded for the language `#1'. Using the pattern for}%
\typeout{** the default language instead.}%
\else
\language=\csname l@#1\endcsname
\fi
#2}}
\providecommand{\BIBdecl}{\relax}
\BIBdecl

\bibitem{wei2023carrier}
Z.~Wei, H.~Liu, X.~Yang, W.~Jiang, H.~Wu, X.~Li, and Z.~Feng, ``Carrier aggregation enabled integrated sensing and communication signal design and processing,'' \emph{IEEE Transactions on Vehicular Technology}, Oct 2023.

\bibitem{ISAC_jwj_1}
Z.~Wei, W.~Jiang, Z.~Feng, H.~Wu, N.~Zhang, K.~Han, R.~Xu, and P.~Zhang, ``{Integrated Sensing and Communication enabled Multiple Base Stations Cooperative Sensing Towards 6G},'' \emph{IEEE Network}, Oct 2023.

\bibitem{ISAC_lht}
H.~Liu, Z.~Wei, F.~Li, Y.~Lin, H.~Qu, H.~Wu, and Z.~Feng, ``Integrated sensing and communication signal processing based on compressed sensing over unlicensed spectrum bands,'' \emph{IEEE Transactions on Cognitive Communications and Networking}, 2024.

\bibitem{10273259}
N.~Huang, M.~Dai, Y.~Wu, L.~Qian, J.~Gao, and Z.~Su, ``{Networked Integration of Sensing and Communication for Extended Reality: Framework, Challenges, and Solutions},'' \emph{IEEE Network}, vol.~38, no.~3, pp. 269--276, 2024.

\bibitem{dou2024integrated}
C.~Dou, N.~Huang, Y.~Wu, L.~Qian, Z.~Shi, and T.~Q. Quek, ``{Integrated Sensing and Communication Enabled Multi-Device Multi-Target Cooperative Sensing: A Fairness-aware Design},'' \emph{IEEE Internet of Things Journal}, 2024.

\bibitem{ni2021uplink}
Z.~Ni, J.~A. Zhang, X.~Huang, K.~Yang, and J.~Yuan, ``{Uplink sensing in perceptive mobile networks with asynchronous transceivers},'' \emph{IEEE Transactions on Signal Processing}, vol.~69, pp. 1287--1300, Feb 2021.

\bibitem{liu2022integrated}
F.~Liu, Y.~Cui, C.~Masouros, J.~Xu, T.~X. Han, Y.~C. Eldar, and S.~Buzzi, ``{Integrated sensing and communications: Toward dual-functional wireless networks for 6G and beyond},'' \emph{IEEE journal on selected areas in communications}, vol.~40, no.~6, pp. 1728--1767, Mar 2022.

\bibitem{wei2023integrated}
Z.~Wei, H.~Qu, Y.~Wang, X.~Yuan, H.~Wu, Y.~Du, K.~Han, N.~Zhang, and Z.~Feng, ``{Integrated sensing and communication signals towards 5G-A and 6G: A survey},'' \emph{IEEE Internet of Things Journal}, Jan 2023.

\bibitem{xiao2023novel}
Z.~Xiao, R.~Liu, M.~Li, and Q.~Liu, ``{A Novel Joint Angle-Range-Velocity Estimation Method for MIMO-OFDM ISAC Systems},'' \emph{arXiv preprint arXiv:2308.03387}, 2023.

\bibitem{li2023reliability}
X.~Li, J.~Zhang, C.~Han, W.~Hao, M.~Zeng, Z.~Zhu, and H.~Wang, ``{Reliability and security of CR-STAR-RIS-NOMA assisted IoT networks},'' \emph{IEEE Internet of Things Journal}, Dec 2023.

\bibitem{li2023iq}
X.~Li, H.~Qi, D.-T. Do, Z.~Hui, Y.~Ding, M.~Zhu, and H.~Peng, ``{IQ-Impaired Wireless-Powered Modify-and-Forward Relaying for IoT Networks: An In-Depth Physical Layer Security Analysis},'' \emph{IEEE Internet of Things Journal}, Feb 2023.

\bibitem{wei2024deep}
Z.~Wei, H.~Liu, Z.~Feng, H.~Wu, F.~Liu, Q.~Zhang, and Y.~Du, ``Deep cooperation in isac system: Resource, node and infrastructure perspectives,'' \emph{IEEE Internet of Things Magazine}, 2024.

\bibitem{zhang2021enabling}
J.~A. Zhang, M.~L. Rahman, K.~Wu, X.~Huang, Y.~J. Guo, S.~Chen, and J.~Yuan, ``{Enabling joint communication and radar sensing in mobile networks—A survey},'' \emph{IEEE Communications Surveys \& Tutorials}, vol.~24, no.~1, pp. 306--345, Oct 2021.

\bibitem{ISAC_jwj_2}
W.~Jiang, Z.~Wei, B.~Li, Z.~Feng, and Z.~Fang, ``{Improve Radar Sensing Performance of Multiple Roadside Units Cooperation via Space Registration},'' \emph{IEEE Transactions on Vehicular Technology}, vol.~71, no.~10, pp. 10\,975--10\,990, Oct. 2022.

\bibitem{liu2024carrier}
H.~Liu, Z.~Wei, J.~Piao, H.~Wu, X.~Li, and Z.~Feng, ``{Carrier Aggregation Enabled MIMO-OFDM Integrated Sensing and Communication},'' \emph{arXiv preprint arXiv:2405.10606}, 2024.

\bibitem{li2019residual}
X.~Li, J.~Li, Y.~Liu, Z.~Ding, and A.~Nallanathan, ``{Residual transceiver hardware impairments on cooperative NOMA networks},'' \emph{IEEE Transactions on Wireless Communications}, vol.~19, no.~1, pp. 680--695, Oct 2019.

\bibitem{xianrong2020research}
W.~Xianrong, Y.~Jianxin, Z.~Weijie, X.~Deqiang, S.~Kan, S.~Jiale, C.~Feng, R.~Yunhua, G.~Ziping, and K.~Hengyu, ``Research progress and development trend of the multi-illuminator-based passive radar,'' \emph{Journal of Radars}, vol.~9, no.~6, pp. 939--958, Dec 2020.

\bibitem{zhang2022integration}
J.~A. Zhang, K.~Wu, X.~Huang, Y.~J. Guo, D.~Zhang, and R.~W. Heath, ``{Integration of radar sensing into communications with asynchronous transceivers},'' \emph{IEEE Communications Magazine}, vol.~60, no.~11, pp. 106--112, Aug 2022.

\bibitem{younis2006performance}
M.~Younis, R.~Metzig, and G.~Krieger, ``Performance prediction of a phase synchronization link for bistatic sar,'' \emph{IEEE Geoscience and Remote Sensing Letters}, vol.~3, no.~3, pp. 429--433, Jul 2006.

\bibitem{jin2019advanced}
G.~Jin, K.~Liu, D.~Liu, D.~Liang, H.~Zhang, N.~Ou, Y.~Zhang, Y.~Deng, C.~Li, and R.~Wang, ``{An advanced phase synchronization scheme for LT-1},'' \emph{IEEE Transactions on Geoscience and Remote Sensing}, vol.~58, no.~3, pp. 1735--1746, Nov 2019.

\bibitem{cai2022advanced}
Y.~Cai, R.~Wang, W.~Yu, D.~Liang, K.~Liu, H.~Zhang, and Y.~Chen, ``{An advanced approach to improve synchronization phase accuracy with compressive sensing for LT-1 bistatic spaceborne SAR},'' \emph{Remote Sensing}, vol.~14, no.~18, p. 4621, Sep 2022.

\bibitem{qian2018widar2}
K.~Qian, C.~Wu, Y.~Zhang, G.~Zhang, Z.~Yang, and Y.~Liu, ``{Widar2. 0: Passive human tracking with a single Wi-Fi link},'' in \emph{Proceedings of the 16th annual international conference on mobile systems, applications, and services}, Jun 2018, pp. 350--361.

\bibitem{zeng2019farsense}
Y.~Zeng, D.~Wu, J.~Xiong, E.~Yi, R.~Gao, and D.~Zhang, ``{FarSense: Pushing the range limit of WiFi-based respiration sensing with CSI ratio of two antennas},'' \emph{Proceedings of the ACM on Interactive, Mobile, Wearable and Ubiquitous Technologies}, vol.~3, no.~3, pp. 1--26, Sep 2019.

\bibitem{weiss2009direct}
A.~J. Weiss and A.~Amar, ``{Direct geolocation of stationary wideband radio signal based on time delays and Doppler shifts},'' in \emph{2009 IEEE/SP 15th Workshop on Statistical Signal Processing}.\hskip 1em plus 0.5em minus 0.4em\relax IEEE, Oct 2009, pp. 101--104.

\bibitem{ren2021improved}
M.~Ren, P.~He, and J.~Zhou, ``{Improved Shape-Based Distance Method for Correlation Analysis of Multi-Radar Data Fusion in Self-Driving Vehicle},'' \emph{IEEE Sensors Journal}, vol.~21, no.~21, pp. 24\,771--24\,781, Sep 2021.

\bibitem{jiang2021research}
W.~Jiang, Z.~Qi, Z.~Ye, Y.~Wan, and L.~Li, ``{Research on Cooperative Detection Technology of Networked Radar Based on Data Fusion},'' in \emph{2021 2nd China International SAR Symposium (CISS)}.\hskip 1em plus 0.5em minus 0.4em\relax IEEE, Dec 2021, pp. 1--5.

\bibitem{wei2023symbol}
Z.~Wei, R.~Xu, Z.~Feng, H.~Wu, N.~Zhang, W.~Jiang, and X.~Yang, ``{Symbol-level Integrated Sensing and Communication enabled Multiple Base Stations Cooperative Sensing},'' \emph{IEEE Transactions on Vehicular Technology}, Aug 2023.

\bibitem{liu2024target}
H.~Liu, Z.~Wei, F.~Yang, H.~Wu, K.~Han, and Z.~Feng, ``{Target Localization with Macro and Micro Base Stations Cooperative Sensing},'' \emph{IEEE Global Communications Conference (GLOBECOM), Accepted}, 2024.

\bibitem{venkatesh2007non}
S.~Venkatesh and R.~Buehrer, ``{Non-line-of-sight identification in ultra-wideband systems based on received signal statistics},'' \emph{IET Microwaves, Antennas \& Propagation}, vol.~1, no.~6, pp. 1120--1130, Dec 2007.

\bibitem{wang2009new}
Z.~Wang, W.~Xu, and S.~A. Zekavat, ``{A new multi-antenna based los-nlos separation technique},'' in \emph{2009 IEEE 13th Digital Signal Processing Workshop and 5th IEEE Signal Processing Education Workshop}.\hskip 1em plus 0.5em minus 0.4em\relax IEEE, Feb 2009, pp. 331--336.

\bibitem{chouchane2009defending}
A.~Chouchane, S.~Rekhis, and N.~Boudriga, ``Defending against rogue base station attacks using wavelet based fingerprinting,'' in \emph{2009 IEEE/ACS International Conference on Computer Systems and Applications}.\hskip 1em plus 0.5em minus 0.4em\relax IEEE, Jun 2009, pp. 523--530.

\bibitem{wu2022joint}
K.~Wu, J.~A. Zhang, and Y.~J. Guo, \emph{Joint communications and sensing: From fundamentals to advanced techniques}.\hskip 1em plus 0.5em minus 0.4em\relax John Wiley \& Sons, Dec 2022.

\bibitem{sturm2011waveform}
C.~Sturm and W.~Wiesbeck, ``{Waveform design and signal processing aspects for fusion of wireless communications and radar sensing},'' \emph{Proceedings of the IEEE}, vol.~99, no.~7, pp. 1236--1259, May 2011.

\bibitem{pasya2014joint}
I.~Pasya, N.~Iwakiri, and T.~Kobayashi, ``{Joint direction-of-departure and direction-of-arrival estimation in an ultra-wideband MIMO radar system},'' in \emph{2014 IEEE Radio and Wireless Symposium (RWS)}.\hskip 1em plus 0.5em minus 0.4em\relax IEEE, Jun 2014, pp. 52--54.

\bibitem{zhang2010direction}
X.~Zhang, L.~Xu, L.~Xu, and D.~Xu, ``{Direction of departure (DOD) and direction of arrival (DOA) estimation in MIMO radar with reduced-dimension MUSIC},'' \emph{IEEE communications letters}, vol.~14, no.~12, pp. 1161--1163, Nov 2010.

\bibitem{tikhonov1963solution}
A.~N. Tikhonov, ``Solution of incorrectly formulated problems and the regularization method.'' \emph{Sov Dok}, vol.~4, pp. 1035--1038, 1963.

\bibitem{hassan2017novel}
A.~Y. Hassan, ``{A novel structure of high speed OFDM receiver to overcome ISI and ICI in Rayleigh fading channel},'' \emph{Wireless Personal Communications}, vol.~97, no.~3, pp. 4305--4325, 2017.

\bibitem{sturm2009ofdm}
C.~Sturm, T.~Zwick, and W.~Wiesbeck, ``{An OFDM system concept for joint radar and communications operations},'' in \emph{VTC Spring 2009-IEEE 69th Vehicular Technology Conference}.\hskip 1em plus 0.5em minus 0.4em\relax IEEE, Jun 2009, pp. 1--5.

\bibitem{trefethen2022numerical}
L.~N. Trefethen and D.~Bau, \emph{{Numerical linear algebra}}.\hskip 1em plus 0.5em minus 0.4em\relax Siam, Jun 2022, vol. 181.

\bibitem{richards2010noncoherent}
M.~A. Richards, ``{Noncoherent integration gain, and its approximation},'' \emph{Georgia Institute of Technology, Technical Memo}, Jun 2010.

\bibitem{jiang2023cooperation}
W.~Jiang, Z.~Wei, S.~Yang, Z.~Feng, and P.~Zhang, ``Cooperation based joint active and passive sensing with asynchronous transceivers for perceptive mobile networks,'' \emph{IEEE Transactions on Wireless Communications}, 2024.

\end{thebibliography}

\end{document}